\documentclass{article}
\usepackage[ruled,linesnumbered]{algorithm2e}
\usepackage{dsfont}
\usepackage{graphicx}
\usepackage{booktabs}
\usepackage{ragged2e}
\usepackage{amsmath}
\usepackage{amsfonts}
\usepackage{amssymb}
\usepackage{amsthm}
\usepackage[utf8]{inputenc}
\usepackage[english]{babel}
\newtheorem{defn}{definition}
\newtheorem{theo}{Theorem}
\newtheorem{lem}{Lemma}
\usepackage{geometry}
\date{}
\begin{document}
\title{Compact Neighborhood Index for Subgraph Queries in Massive Graphs}
\author{C. Nabti*, T. Mecharnia, S. E. Boukhetta*, H. Seba* and K. Amrouche** \\
*Universit\'{e} de Lyon, CNRS, Universit\'{e} Lyon 1\\
 LIRIS, UMR5205, F-69622 Lyon, France.\\
hamida.seba@univ-lyon1.fr\\
** Ecole Supérieure d'Informatique(ESI, Oued smar, Algérie)\\
}

\maketitle
\begin{abstract}
Subgraph queries also known as subgraph isomorphism search is a fundamental problem in querying graph-like structured data. It consists to enumerate the subgraphs of a data graph that match a query graph. This problem arises in many real-world applications related to query processing or pattern recognition such as computer vision, social network analysis, bioinformatic and big data analytic. Subgraph isomorphism search knows a lot of investigations and solutions mainly because of its importance and use but also because of its NP-completeness. Existing solutions use filtering mechanisms and optimise the order within witch the query vertices are matched on the data vertices to obtain acceptable processing times. However, existing approaches
are iterative and generate several intermediate results.
They also require that the data graph is loaded in main memory and consequently are not adapted to large graphs that do not fit into memory or are accessed by streams.  To tackle this problem, we propose a new approach based on concepts widely different from existing works. Our approach distills the semantic and topological information that surround a vertex into a simple integer. This simple vertex encoding that can be computed and updated incrementally reduces considerably intermediate results and avoid loading the entire data graph into main memory. We evaluate our approach on several  real-word datasets. The experimental results show that our approach is efficient and scalable.
\end{abstract}

\section{Introduction}

\label{Sec-Introduction}
Graphs are not a new paradigm for data representation and modeling. Their use in these domains dates  back to the birth of computer databases with, for example, the work of Bachman on the \textit{Network} database model \cite{Bachman1969}.
However, the advent of applications related to nowadays connected world with social networks, online crime detection, genome and scientific databases, etc., has brought graphs to greater prominence. This is due mainly to their adaptability to represent the linked aspect of nowadays data but also to their flexibility and scalability when dealing with the main challenge of these kind of applications : Massive data. In this context, subgraph isomorphism search is a fundamental task on which are based search and querying algorithms. Subgraph isomorphism search, also known as exact subgraph matching or subgraph queries, is the problem of enumerating all the occurrences of a query graph within a larger graph called the data graph (cf. Figure \ref{Fig-Example}).
\begin{figure*}
\centering
\includegraphics[scale=0.25]{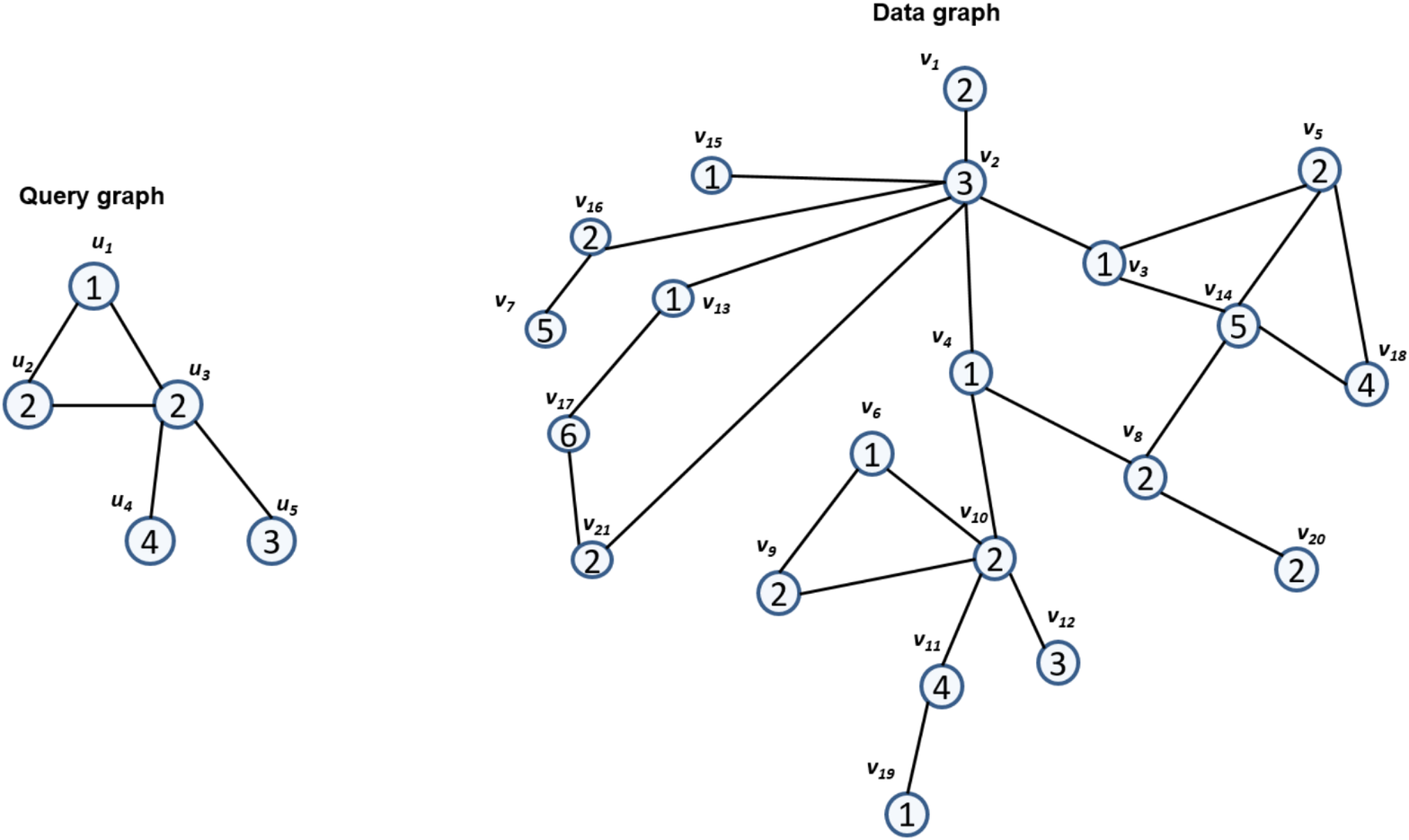}
\caption{Running Example.\label{Fig-Example}}
\end{figure*}
 Subgraph isomorphism search  is an NP-complete problem that knows extensive investigations and solutions mainly because of its importance and use, but also because known solutions are memory and space expensive and consequently do not scale well \cite{Lee2013,Gallagher2006,Shasha2002}. Most existing solutions are extensions of the well known Ullmann's algorithm \cite{Ullmann76}. These solutions are based on exploring a search space in the form of a recursion tree that maps the query vertices to the data graph's vertices. However, parsing the recursion tree is exponential in function of the number of vertices in the involved graphs. So, existing solutions never construct entirely the recursion tree and use pruning methods to have smaller search spaces.  The most referenced methods such as Ullmann's algorithm \cite{Ullmann76},  VF2 \cite{Cordella2004}, QuickSI \cite{Shang2008}, GraphQL \cite{He2008},  GADDI \cite{Zhang2009},  SPath \cite{ZhaoP2010}, Turbo$_{ISO}$ \cite{Han2013} and CFL-Match \cite{Bi2016} have different approaches to tackle this problem. These solutions are surveyed in several papers \cite{Katsarou2017,Lee2013,Gallagher2006,Shasha2002} that analyse and compare them. These comparative studies agree on the fact that there is not a method that outperforms the others for all kind of queries and all graph classes. It seems that the exiting solutions are complementary. This motivated the authors of \cite{Katsarou2017} to propose a framework that allows to run several algorithms on parallel threads on the same query.

  In this work, we analyse existing algorithms in regard of the cost of their filtering task and how this task is positioned in respect to the whole matching process. To our knowledge, this issue has not been investigated before. In fact, existing algorithms are generally built around two main tasks: filtering and searching. Filtering is fundamental as it reduces the search space explored by the searching task. Existing algorithms differ by the pruning power of the filtering mechanisms they implement but also by when these filters take place with respect to searching. Our analysis of these two points of difference highlighted four weaknesses in the state of the art algorithms that we address within the proposed framework.

 \textit{ Weakness 1:} \textit{High filtering cost}. The main pruning mechanism used by existing methods during filtering is the features of the
$k-$neighborhood of query vertices. This is the amount of information used when matching a query vertex with data vertices. The more information is used, i.e., $k$ is big, the more the pruning of the search space can be important. However, representing compactly the $k-$neighborhood for practical comparisons is a challenging issue. In fact, the representation of this information has a direct impact on its cost which increases with the value of $k$.
Besides filtering with the vertex label and the vertex degree, the lightest $k$-neighborhood filter is to consider the features of the one-hop neighborhood, i.e., $k=1$. For this, recent approaches such as Turbo$_{Iso}$ \cite{Han2013} and CFL-Match \cite{Bi2016} use the  Neighborhood Label Frequency (NLF) filter \cite{Zhu2012}. NLF ensures that a data vertex $v$ is a candidate for a query vertex $u$ only if the neighborhood of $v$, denoted $N(v)$, includes the neighborhood of $u$ (see lines 5-9 of  Algorithm \ref{Algo-NLF}).

\begin{algorithm}
\KwData{A potential candidate vertex $v$ for a query vertex $u$ }
\KwResult{TRUE if $v$ is candidate for $u$ and FALSE otherwise}
\Begin{
\If{$mnd_G(v)<mnd_Q(u)$}
{\Return (FALSE)\;}
\ForEach{ label $l \in \ell(N(u))$ }{
\If{$|\{w \in N(v) |\ell(w)=l\}|< |\{w \in N(u) |\ell(w)=l\}|$}
{\Return(FALSE)\;}
}
\Return(TRUE)\;
}
\caption{NLF and MND filters.\label{Algo-NLF}}
\end{algorithm}

However, NLF is expensive: it is $\mathcal{O}(|V(Q)||V(G)||\mathcal{L}(Q)|)$ where $|V(Q)|$ is the number of vertices of the query, $|V(G)|$ is the number of vertices in the data graph and $\mathcal{L}(Q)$ is the set of unique labels of the query graph which is $\mathcal{O}(|V(Q)|)$ in the worst case. So, to avoid applying NLF systematically on each vertex, CFL-Match \cite{Bi2016} proposes  the Maximum Neighbor-Degree (MND) filter, which can be verified in constant time for each candidate data vertex, i.e., $\mathcal{O}(|V(Q)||V(G)|)$. The maximum neighbor-degree of a vertex $u$ in a
graph $G$, denoted $mnd_G(u)$, is the maximum degree of all its neighbors \cite{Bi2016}. A data vertex $v$ is not a candidate for a query vertex $u$ if $mnd_G(v)<mnd_Q(u)$. As MND is not as powerful as NLF, the idea is to apply it before applying NLF as detailed in  Algorithm \ref{Algo-NLF} (see lines 2-3). However, MND is not always effective as we can see in the example depicted in Figure \ref{Fig-MND} where only 3 vertices are pruned with the MND filter and consequently NLF must be applied for each of the remaining vertices.

\begin{figure}
\centering
\includegraphics[scale=0.25]{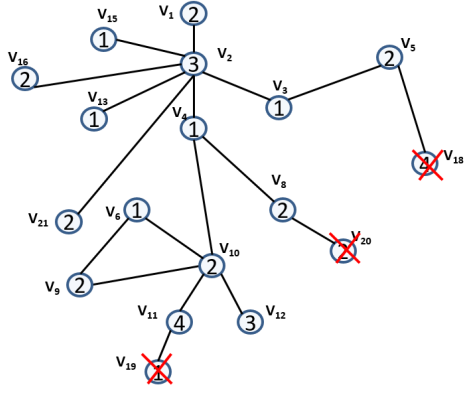}
\caption{MND Filter on on the Running Example (pruned of the vertices that do not match query labels).\label{Fig-MND}}
\end{figure}
It is also worth noting that for some neighborhood configurations filtering is useless and only the searching step is decisive. Let consider the query and data graphs depicted in Figure \ref{Fig-NLFNeedless} where all the vertices have the same label and the same degree and let consider that $k=1000$. Clearly, in this case, the 1000 comparisons required by NLF for each query vertex and each data vertex are needless. This doesn't mean that filtering is not necessary but that its cost must be reduced. Interestingly, using a less costly filtering with Ullmann's native subgraph searching subroutine outperforms the state of the art algorithms as showed by our experiments.

\begin{figure}
\centering
\includegraphics[scale=0.25]{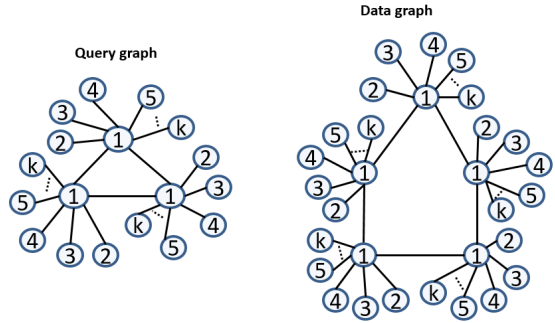}
\caption{Needless NLF filtering\label{Fig-NLFNeedless}}
\end{figure}

\textit{Weakness 2: Global filtering Vs local filtering}.
Depending on its scope, filtering can be characterised as global or local. Local filtering designates the filtering methods that reduce the number of data vertices candidates for a given query vertex, i.e., reduce the size of $C(u_i)$,  $i=1, |V_Q|$, where  $C(u_i)$ is the set of vertices of the data graph that are candidates for the query vertex $u_i$. Global filtering designates the filtering methods that can be applied on the entire search space, obtained by joining the above sets, i.e., $C(u_1)\times C(u_2)\times \cdots\times C(u_{|V_Q|})$. Our study of existing algorithms shows that local pruning is predominant.
 Some mechanisms allow global pruning but they require extra passes of the data graph to be effective. The matching order is such a mechanism. However, it is a very difficult problem to choose a robust matching order mainly because the number of all possible matching orders is exponential in the number of vertices. So, it is expensive to enumerate all of them. For example, Tuorbo$_{Iso}$ relies on vertex ordering for pruning. However, to compute this order, it needs to compute for each query vertex a selectivity criteria based on the frequency of its label in the data graph. This means, parsing the data graph for computing these frequencies.

 To deal with this problem, we introduce the Iterative Local Global Filtering mechanism (ILGF),  a simple way to achieve global punning relying on local pruning filters.

 \textit{Weakness 3: Late filtering}.
 Our analysis of how filtering and searching  are undertaken with respect to each other in the state of the art algorithms revealed that most algorithms apply their filtering mechanisms during subgraph search. In fact, little filtering, reduced mainly to label or degree filtering, is undertaken prior to subgraph search. This means that, the first cartesian products involved in the subgraph search task are costly. To tackle this, CFL-Match \cite{Bi2016} applies the MND-NLF filter, locally,  prior to subgraph search. However, as we can see in Figure \ref{Fig-NLFOrder}, the amount of achieved pruning depends on the order within which vertices are parsed. In our example, if $v_2$ is processed before $v_{16}$ the amount of pruning is less than the one obtained with the reverse order.
 \begin{figure*}
\centering
\includegraphics[scale=0.23]{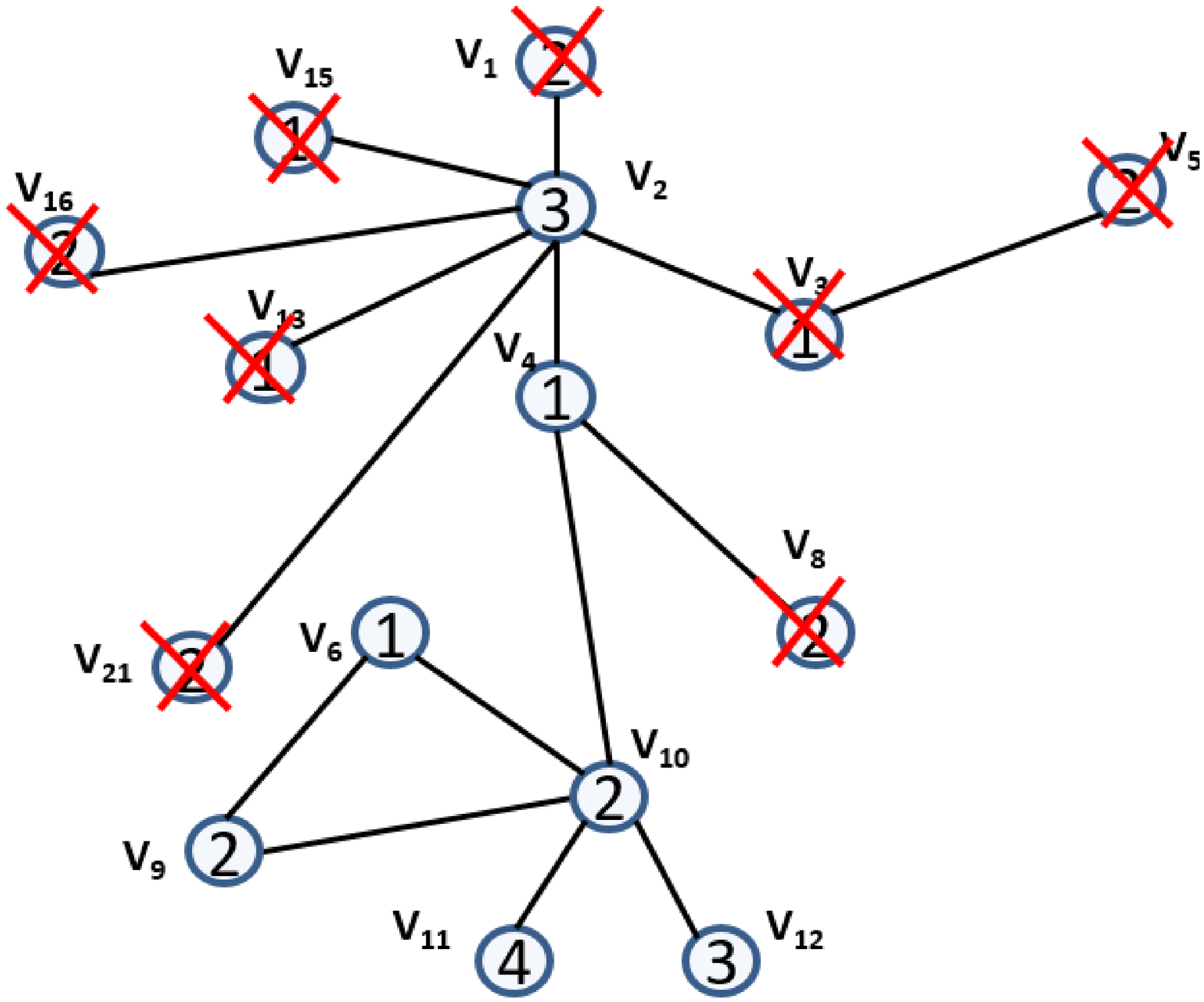}\hspace{4cm}
\includegraphics[scale=0.23]{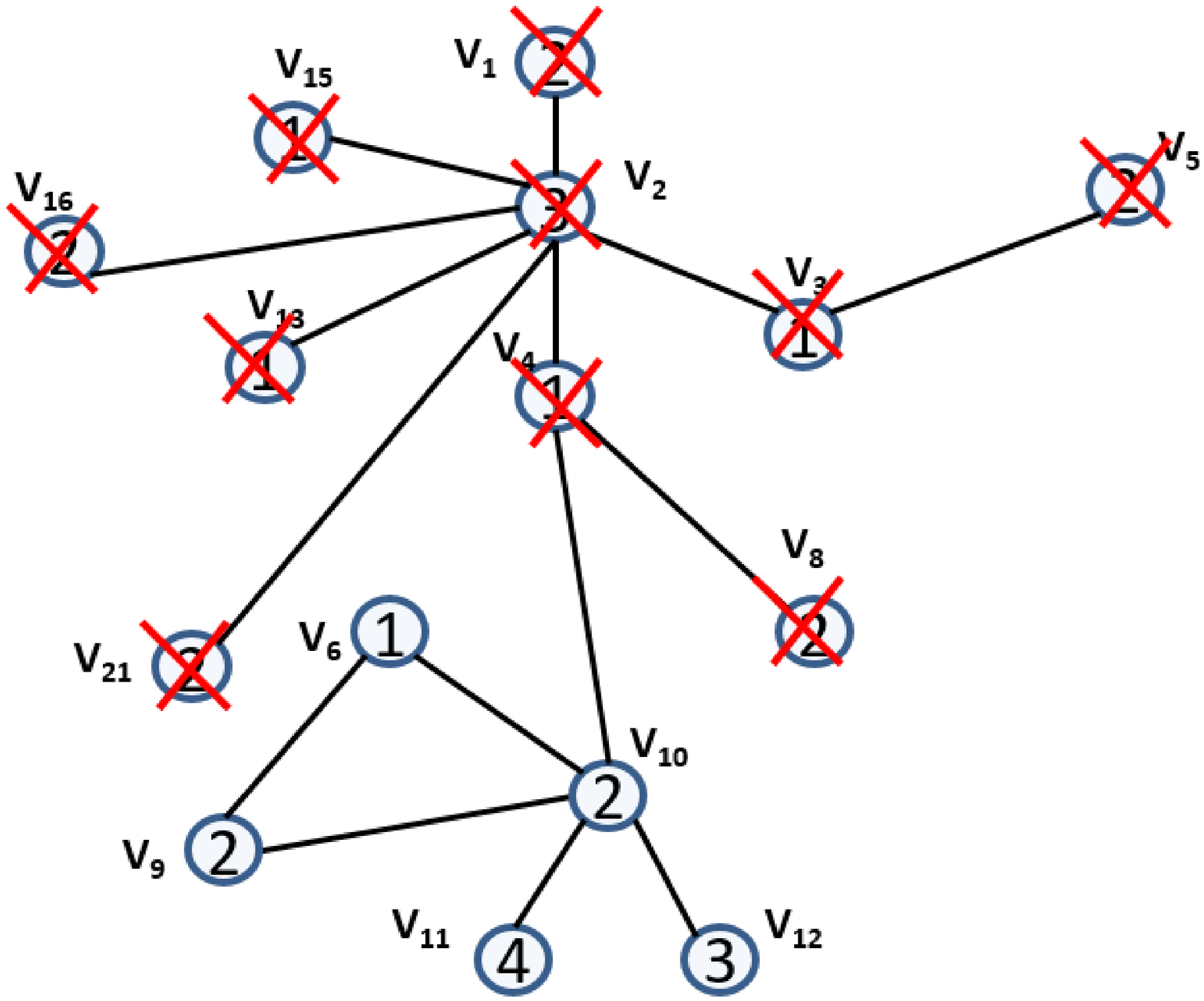}\\

$v_2$ is processed before $v_{16}$ \hspace{4cm}  $v_{16}$ is processed before $v_2$
\caption{NLF filtering with two different vertex parsing orders\label{Fig-NLFOrder}}
\end{figure*}
To get caught up, existing solutions rely on additional mechanisms and data structures during subgraph search such as \textrm{NEC tree} in Turbo$_{Iso}$ \cite{Han2013} and \textit{CPI} in CFL-Match \cite{Bi2016} that both use path-based ordering during subgraph search. However, the underlying data structures are time and space exponential \cite{Bi2016}. To avoid constructing and maintaining such data structures, we propose to achieve filtering solely prior to subgraph search. Our experiments show that this approach is more efficient than the state of the art algorithms.

\textit{Weakness 4: lack of scalability.} This drawback results directly from the three above weaknesses. In fact, the lack of global filtering and the necessity to keep the data graph into memory for several passes make these backtracking-based solutions  not suitable for graphs that do not fit into main memory. In this work, we aim to achieve a single parse of the data graph and reduce as early as possible the search space. Our contributions are:
\begin{itemize}
\item We propose a novel encoding of vertices, called Compact Neighborhood Index (CNI) that distills all the information around a vertex in a single integer leading to a simple but extremely efficient filtering scheme for processing subgraph isomorphism search. The whole filtering process is based on integer comparisons. CNIs are also easily updatable during filtering.
\item We propose an Iterative Local Global filtering algorithm (ILGF) that relies on the characteristics of CNIs to ensure a global pruning of the search space before subgraph search.
\item Our encoding mechanism has the advantage to adapt to all graph access models: main memory, external memory and streams. By performing one sequential pass of
the disk file (or the stream of edges) of the input graph. This avoids expensive random disk accesses if the graph does not fit into main memory.
\item We conduct extensive experiments using synthetic and real datasets in different application domains to attest the effectiveness and efficiency of the proposed scheme.
\end{itemize}

The rest of this paper is organized as follows. Section \ref{Sec-Background} first
formalizes the problem of subgraph isomorphism search and defines the notations used throughout the paper, then, it discusses related work to motivate our contribution. In Section
\ref{Sec-Approach}, we introduce our main contribution: the compact neighborhood index and how it is used to solve subgraph isomorphism search. Section \ref{Sec-Evaluation} presents a comprehensive experimental study on several datasets. Section \ref{Sec-Conclusion} concludes the paper.

\section{Background}
\label{Sec-Background}
\subsection{Problem Definition}
We use a data graph to represent objects and their relationships using vertices and edges. More formally,
a data graph $G$ is a 4-tuple $G =(V(G), E(G),\ell, \Sigma)$, where $V(G)$ is a set of vertices (also called nodes), $E(G)\subseteq V(G)\times V(G)$ is a set of edges connecting the vertices, $\ell: V(G) \cup E(G)\rightarrow \Sigma$ is a labeling function on the vertices and the edges  where $\Sigma$ is the set of labels that can appear on the vertices and/or the edges.
We use $|V(G)|$ and $|E(G)|$ to represent respectively the number of vertices and the number of edges in $G$.

An undirected edge between vertices $u$ and $v$ is denoted indifferently by $(u, v)$ or $(v, u)$. A neighbor of a vertex $v$ is a vertex adjacent to $v$. The degree of a vertex $v$, denoted $deg(v)$, is the number of its neighbors. We also use $deg_S(v)$ to denote the number of neighbors of $v$ that have a label in the set $S$. We use $N(v)$ to represent the neighbors of vertex $v$.
$\ell_G(u)$ (or simply $\ell(u)$ when there is no ambiguity) represents the label of vertex $u$ in $G$ and  $\ell((u,v))$ is the
label of the edge $(u,v)$ in $G$.

A graph that is contained in another graph is called a subgraph and can be defined as follows:
\begin{defn}
A graph $G_1 = (V(G_1),E(G_1), \ell_1, \Sigma)$ is a subgraph of a graph $G_2 = (V(G_2),E(G_2), \ell_2, \Sigma)$ if  $V(G_1)\subseteq V(G_2)$, $E(G_1)\subseteq E(G_2)$,
$\ell_{1}(x)=\ell_{2}(x)$  $\forall x \in V(G_1)$, and $\ell_{1}(e)=\ell_{2}(e)$  $\forall e \in E(G_1)$.
\end{defn}

\begin{defn}
A graph $Q = (V(Q),E(Q), \ell, \Sigma)$ is subgraph isomorphic to a graph $G = (V(G),E(G), \ell, \Sigma)$ if and only if there
exists an injective mapping $h$ from $V(Q)$ to $V(G)$ such that :
\begin{enumerate}
\item $\forall x \in V(Q): \ell(x) = \ell(h(x))$
\item $\forall(x,y)\in E(Q): (h(x), h(y))\in E(G)$
\end{enumerate}
\end{defn}

For presentation convenience, we do not show edge labels on our examples but these labels are considered in our algorithms and datasets.
Table \ref{Table-notation} summarises our notation.
\begin{table}[!h]
    \centering
    \caption{Notation}
    \label{Table-notation}
\begin{tabular}{ll}
\toprule
\textbf{Symbol} & \textbf{Description} \\
\midrule
$G = (V,E,\ell,\Sigma)$& undirected vertex and edge labeled graph\\
&$\ell$ is a labeling function \\
&$\Sigma$ is the set of labels\\
$V(G)$ & vertex set of the graph $G$ \\
$E(G)$ & edge set of the graph $G$ \\
$deg(v)$& degree of vertex $v$ in $G$\\
$deg_S(v)$& number of neighbors of $v$\\
& that have a label in $S$\\
$G[X]$ & the  subgraph of $G$ induced by the set  \\
& of vertices $X$\\
$\mathcal{L}(Q)$& the set of unique labels in the query $Q$\\
$cni(v)$& compact neighborhood index of $v$\\
\bottomrule

\end{tabular}
 \end{table}
\subsection{Related Work}
\label{Sec-RelatedWork}
Many algorithms are proposed to solve the subgraph isomorphism search problem. We can cite without being exhaustive Ullmann's algorithm \cite{Ullmann76}, VF2 \cite{Cordella2004}, QuickSI \cite{Shang2008}, GraphQL \cite{He2008},  GADDI \cite{Zhang2009},  SPath \cite{ZhaoP2010}, Turbo$_{ISO}$ \cite{Han2013} and CFL-Match \cite{Bi2016}. One can find in  \cite{Lee2013} and \cite{Katsarou2017} useful studies that survey and compare most of these methods on several aspects of query processing. These surveys mainly show that none of existing methods is adapted for all kind of queries and graphs and that existing algorithms are somewhat complementary. So, in this section, we review the existing solutions on other facets that outline and justify our contributions. Existing algorithms for subgraph isomorphism search are built onto two basic tasks: Filtering and Searching. Filtering is an important step and determines the efficiency of the algorithm. The searching step is generally based on the Ullmann's backtracking subroutine \cite{Ullmann76} that searches in a depth-first manner for matchings between the query graph and the data graph obtained by the filtering step.  So, the aim of the filtering step is to reduce the search space on which the searching step operates.
Filtering mechanisms can be classified into two categories depending on their scope: local or global. A local filtering mechanism prunes the set of mappings that are candidates for a single vertex. A global pruning operates on the whole search space.
\begin{itemize}
\item Global pruning: as global pruning mechanisms, we can cite vertex ordering and query rewriting. Vertex ordering is the selection of an order within which the vertices of the query are handled. In fact, this order has a direct incidence on the size of the search space as demonstrated by several examples \cite{Lee2013}. Query rewriting consists in representing the query in a form that simplifies its matching.
Ullmann's algorithm and SPath do not define any global pruning mechanism and picks the query vertices in a random manner. VF2 and GADDI handle a query vertex only if it is connected to an already matched vertex. However, GADDI uses an additional mechanism : a distance based on the number of frequent substructures between the $k-$neighborhoods of two vertices as a mean to prune globally the search space after each established mappings between a query vertex and a data vertex. QuickSI rewrites the query in the form of a tree: a spanning tree of the query. Edges and vertices of the query are weighted by the frequency of their occurrence in the data graph. Based on these weights, a minimum spanning tree is constructed and used to search the data graph.  GraphQL selects the vertex that minimises the cost of the ongoing join operation. The cost of a join is estimated by the size of the product of the involved sets of vertices.
 Turbo$_{ISO}$ uses the ordering introduced in  \cite{Zhao2012}. This ordering uses the popularity of query vertices in the data graph. Every query vertex $u$ is ranked by $rank(u) = \frac{freq(G,\ell(u))}{deg(u)}$  where $freq(G, l)$ is the number of data vertices in $G$ that have label $l$.  This ranking function favors lower frequencies and higher degrees which will minimize the number of candidate vertices in the data graph. Furthermore, Turbo$_{ISO}$ rewrites the query within a tree using this ranking as in QuickSI but Turbo$_{ISO}$ aggregates the vertices that have the same labels and the same neighbors into a single vertex. This aggregation has been extended to data graphs in \cite{Ren2015}. A similar, but more general, compressing-based approach, called Sum$_{ISO}$ \cite{Lagraa2014,Nabti2017} uses modular decomposition of graphs \cite{Gallai1967} to aggregate vertices that have the same neighborhood into supervertices and undertakes subgraph search on the compressed graphs.
    More recently, the authors of \cite{Bi2016} claim that spanning trees are not the best solutions to represent queries. They show that the edges not included in the spanning tree may have an important pruning power. So, they propose to enhance the tree representation by partitioning the query graph into a core and a forest.

\item Local pruning consists mainly in reducing the number of mappings available for each query vertex. In fact, the final search space is the result of joining the sets of available  mappings of each query vertex. Thus, given a query graph $Q=(V(Q),E(Q),\ell,\Sigma)$ and a data graph $G=(V(G),E(G),\ell,\Sigma)$, the aim is to reduce as much as possible the sets $C(u_i)$,  $i=1, |V_Q|$, where  $C(u_i)$ is the set of vertices of the data graph that match the query vertex $u_i$. The final search space is  obtained by joining these sets, i.e., $C(u_1)\times C(u_2)\times \cdots\times C(u_{|V_Q|})$ \cite{He2008}.
 The reduction of $C(u)$ is generally achieved using the neighborhood information of $u$. The amount of the obtained pruning depends on the scope of the considered neighborhood. The simplest solution considers the one-hop neighborhood such as the degree of the vertex and/or the labels of the neighbors. Neighborhood at $k-$hops is also used in some methods.
 Ullmann's Algorithm refines $C_u$ by removing the vertices that have a smaller degree than $u$. GraphQL also uses the direct neighborhood by encoding within a sequence the labels of the neighbors of each vertex. Furthermore, GraphQL uses an approximation algorithm proposed in \cite{He2006} to further reduce the search space by discarding the data vertices that are not compatible with the query vertex using the $k-$neighborhood around $u$. VF2 looks to 2-hops neighborhood. SPath uses the $k$-neighborhood by maintaining for each vertex $u$ a structure that contains the labels of all vertices that are at a distance less or equal than $k$ from $u$. SPath uses its encoding of the $k$-neighborhood to remove the data vertices that have a $k$-neighborhood that does not englobe any $k$-neighborhood of query vertices.
By rewriting, the query within a tree, QuickSI and Turbo$_{ISO}$ use also the $k-$neighborhood with the particularity that the neighborhood is rooted at a more pruning vertex. The tree representation of Turbo$_{ISO}$ is also more compact as it aggregates similar vertices.
\end{itemize}

From the above discussion of existing methods, one can see that the main concern is reducing the search space. The pillar of such quest is the amount of semantic and topological information we maintain for each vertex and how this information is encoded. The more information we use the more pruning we achieve. However, the encoding of this information has a direct impact of its usefulness in pruning. Moreover, existing methods use static encoding of the vertex neighborhood in that it is not updatable after a local pruning. In this paper, we focus on simplifying the encoding of the $k$-neighborhood so that : (1) to reduce its cost for filtering and (2) to be able to simply update it after each local pruning to ensure a global pruning of the search space as early as possible without the need of complex data structures.

\section{A Novel Approach}
\label{Sec-Approach}
In this  paper, we propose a novel approach to subgraph isomorphism search that aims to reduce the cost of the filtering step. The approach is also adapted for all access methods and especially for big graphs that are accessed within a stream or in external memory.
The main task of the proposed framework is a filtering step that relies on integer comparisons. This step is followed by Ullmann's matching subroutine. The efficiency of the filtering step relies on a novel method to encode a vertex. This encoding distills all the neighboring information that characterise a vertex into a single integer. Unlike existing methods that statically and invariably encode neighboring information, our vertex encoding integer can be dynamically updated leading to an iterative filtering process that allows a global pruning of the search space without additional data structures.

In the following, we first describe this encoding method, called Compact Neighborhood Index, then, we describe the filtering and matching steps of our subgraph matching framework.
\subsection{Compact Neighborhood Index (CNI)}
\label{sec-cni}
In our method the high-level idea is to put into a simple integer the neighborhood information that characterise a vertex. Matching two vertices is then a simple comparison between integers. Given a vertex $u$, the compact neighborhood index of $u$, denoted $cni(u)$, distills the whole structure that surrounds the vertex into a single integer. It is the result of a bijective function that is applied on the vertex's  neighborhood information. This function ensures that two given vertices  $u$ and $v$ will never have the same compact neighborhood index if they have the same label and the same number of neighbors unless they are isomorphic at one-hop.
To compute CNIs, we use pairing functions. A pairing function  on a set $A$ associates each pair of members from $A$ with a single member of $A$, so that any two distinct pairs are associated with two distinct members \cite{Hopcroft2007}. It is a bijection \cite{Stein2013} and according to Fueter–P\'{o}lya theorem \cite{Fueter1923}, the only quadratic pairing function is the Cantor polynomial
    $f: \mathbb{N }\times \mathbb{N} \to \mathbb{N} $ \\
defined by
$f(k_{1},k_{2})={\frac {1}{2}}(k_{1}+k_{2})(k_{1}+k_{2}+1)+k_{2}.$
It assigns consecutive numbers to points along diagonals in the plane.
 To pair more than two numbers, pairings of pairings can be used. For example $f(i,j,k)$ can be defined as $f(i,f(j,k))$ or $f(f(i,j),k)$, but $f(i,j,k,l)$ is defined as $f(f(i,j),f(k,l))$ to minimize the size of the produced number. So, by composing $k-1$ times the bijection of $\mathbb{N}^2$ on $\mathbb{N}$, we  obtains a bijection of $\mathbb{N}^k$ on $\mathbb{N}$ which is a polynomial of degree $k$ \cite{Lisi2007}.
 The expression of the resulting polynomial are obtained thanks to the properties of the Pascal triangle \cite{Hopcroft2007}.

 We use such a generalised pairing function to encode the neighborhood of each vertex.
Let $x_1,x_2,x_3,\cdots, x_k$ be the list of  $u$'s neighbors' labels. The compact neighborhood index of $u$  in the graph $G$ is given by:\\
  $cni(u)=\hbar(1,x_1)+\hbar(2,x_1+x_2)+\cdots+\hbar(k,x_1+x_2+x_3+\cdots+ x_k).$

  So,
  $cni(u)=\sum\mbox{\ensuremath{_{j=1}^{k}\hbar(j,x_1+...+x_j)}}$  where $\hbar(p,s)=\binom{s+p-1}{p}=\frac{(s+p-1)!}{p!(s-1)!}$\\

  \begin{theo} \label{Theo-bijection}
 $\forall (x_1,x_2,x_3,\cdots, x_k)\in \mathds{N}^k$ and $k>0$, $g_{k}$ is a bijective function from $\mathds{N}^{k}$ in $\mathds{N}$, where:\\
  $$g_{k}(x_1,x_2,x_3,\cdots,x_k)=\sum\mbox{\ensuremath{_{j=1}^{k}\hbar(j,x_1+...+x_j)}}$$
  and $$\hbar(p,s)=\binom{s+p-1}{p}=\frac{(s+p-1)!}{p!(s-1)!}$$
\end{theo}
A proof of this Theorem is presented in Appendix \ref{app-prooftheo}.
To use this bijection on vertices' labels, we first assign a unique integer to each vertex label. This assignment can be simply achieved by numbering labels parting from $1$ or by using an associative array to store the query labels. Let $ord(\ell(u))$ be the subroutine used to retrieve the integer associated to the label of vertex $u$. $ord(\ell(u))$ will return $0$ if vertex $u$ has  a label that does not belong to $\mathcal{L}(Q)$. This will systematically prune the neighbors that do not verify the label filter and avoid to consider them in the computation of the CNI of a vertex.

In our case, $k=|\mathcal{L}(Q)|$, i.e., the number of distinct labels in the query. To compute $\hbar(j,x_1+...+x_j)$, $j$ corresponds to a query label index, and $x_j$ is the number of apparition of label $j$ in the direct neighborhood of vertex $u$. This provides to CNIs the same filtering capacity as the NLF filter.\\

\textbf{Example:\\}
Figure \ref{Figure-Example-CNI} illustrates the CNIs of the query graph of our pruning example. These CNIs are computed as follows:\\
For this query graph, the integers used to represent the labels are: 1, 2, 3, and 4, i.e.,  $k=4$.\\
$cni(u_1)=\hbar(1,0)+\hbar(2,0+2)+\hbar(3,0+2+0)+\hbar(4,0+2+0+0)=0+3+4+5=12$. In fact, we can see that labels $1$, $3$ and $4$ do not appear in the neighborhood of of $u_1$ and label $2$ appears $2$ times. The remaining CNIs are computed similarity:\\
$cni(u_2)= g_4(1,1,0,0)=\hbar(1,1)+\hbar(2,1+1)+\hbar(3,1+1+0)+\hbar(4,1+1+0+0)=13$.\\
$cni(u_3)=g_4(1,1,1,1)=\hbar(1,1)+\hbar(2,1+1)+\hbar(3,1+1+1)+\hbar(4,1+1+1+1)=49$.\\
$cni(u_4)=g_4(0,1,0,0)=\hbar(1,0)+\hbar(2,0+1)+\hbar(3,0+1+0)+\hbar(4,0+1+0+0)=3$.\\
$cni(u_5)=g_4(0,1,0,0)=\hbar(1,0)+\hbar(2,0+1)+\hbar(3,0+1+0)+\hbar(4,0+1+0+0)=3$.\\

For filtering, we rely on three filters: the label filter, the degree filter and the CNI filter. The label and degree filters are the basis of all pruning methods. The CNI filter is based on the above bijection. So, we verify candidates for query vertices by the lemmas below.
\begin{lem} [Label filter] \label{lem-label}
Given a query $Q$ and a data graph $G$, a data vertex $v \in V(G)$ is not a candidate of $u \in V(Q)$ if $\ell(v)\neq \ell(u)$.
\end{lem}

\begin{lem} [Degree filter]\label{lem-deg}
Given a query $Q$ and a data graph $G$, a data vertex $v \in V(G)$ is not a candidate of $u \in V(Q)$ if $deg_{\mathcal{L}(Q)}(v)<deg_{\mathcal{L}(Q)}(u)$.
\end{lem}

\begin{lem} [CNI filter]\label{lem-cni}
Given a query $Q$ and a data graph $G$, a data vertex $v \in V(G)$ that verifies the label and degree filters is not a candidate of $u \in V(Q)$ if $cni(v)<cni(u)$.
\end{lem}

Lemmas \ref{lem-label} and \ref{lem-deg} are straightforward. The proof of Lemma \ref{lem-cni} is given in Appendix \ref{app-prooflemcni}.
\begin{figure}
\centering
\includegraphics[scale=0.30]{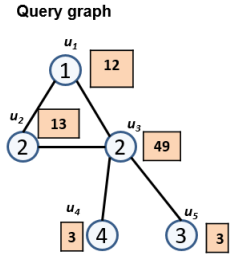}
\caption{CNIs of the Query graph.\label{Figure-Example-CNI}}
\end{figure}

Note that, the CNI filter can be verified in constant time; that is, verifying one candidate vertex $v$
for a query vertex $u$ takes $\mathcal{O}(1)$ time versus  $\mathcal{O}(\mathcal{L}(Q))$ for NLF.

\subsection{Iterative Local Global Filtering Algorithm (ILGF)}
The aim of the Iterative Local Global Filtering Algorithm (ILGF) is to  reduce globally the search space using CNIs. It relies on the fact that $cni(v)$ can be easily updated after a local filtering giving rise to  new filtering opportunities.  Algorithm \ref{Algo_Filtering} details this iterative filtering process. To verify the CNI filter on a candidate data vertex, the algorithm uses the \textit{cniVerify()} subroutine that implements Lemma \ref{lem-cni} and consequently allows to verify that a data vertex  is a candidate for a given query vertex according to the label, degree and CNI filters defined above.
The ILGF algorithm removes iteratively from $G$ the vertices that do not match a query vertex using the label, the degree and the CNI filters (see lines 5-16) of the algorithm. Each time a vertex is removed by the filtering process the degree and CNI of its neighbors are updated (lines 11-14) giving rise to new filtering opportunities. However, it is important to note that filtering iterations do not parse all the remaining vertices in the data graph. In fact, the set \textit{nextFilter} keep track of  the neighbors of the vertices pruned during the current iteration. The next iteration parses only the vertices contained in \textit{nextFilter}. Filtering stops when no further vertices are removed, i.e., \textit{nextFilter} is empty.

\begin{algorithm}
\KwData{A data graph $G$ }
\KwResult{A filtered version of $G$}
\Begin{
stopFilter $\leftarrow$ \textit{FALSE}\;
toFilter $\leftarrow V(G)$\;
nextFilter $\leftarrow \emptyset$\;
\Repeat {stopFilter}{
\ForEach{ vertex $v \in toFilter$}{
  \If {$\forall u \in V(Q), !cniVerify(v,u)$}
         { \textit{affected} $\leftarrow N(v)$\;
           nextFilter $\leftarrow$ nextFilter $\cup N(v)$\;
           remove $v$ from $V(G)$ and the corresponding edges from $E(G)$\;
           \ForEach {$x \in$ \textit{affected} } {
                 update $deg(x)$\;
                 update $cni(x)$\;
                            }
           }
            }
  \If {$|nextFilter|=0$} {stopFilter $\leftarrow$ \textit{TRUE}\;}
  \Else {toFilter $\leftarrow$ nextFilter\;
         nextFilter $\leftarrow \emptyset$\;}
}

\ForEach{ vertex $u\in V(Q)$}{
          $C(u)\leftarrow \{v \in V(G) $ such that $ cniVerify(v, u)\}$\;
          \If {$C(u)=\emptyset$}
          {\Return ($\emptyset$)\;
         }
         }
$M \leftarrow \emptyset$\;
SubgraphSearch($M$)\;
}
\caption{ILGF.\label{Algo_Filtering}}
\end{algorithm}

\begin{algorithm}
\KwData{A data vertex $v$ and a query vertex $u$.}
\KwResult{returns true if $v$ is a candidate for $u$ according to the label, degree and CNI filters.}
\Begin{
\Return $(\ell(u)=\ell(v)\bigwedge deg_{\mathcal{L}(Q)}(u)< deg_{\mathcal{L}(Q)}(v)\bigwedge cni(u)< cni(v))$ or $(\ell(u)=\ell(v)\bigwedge deg_{\mathcal{L}(Q)}(u)= deg_{\mathcal{L}(Q)}(v)\bigwedge cni(u)=cni(v)))$
}
\caption{Function \textit{cniVerify}($v$,$u$).\label{Algo-FootMatch}}
\end{algorithm}

Figure \ref{Figure-DataGraph-Filetering} illustrates the ILGF algorithm on our example. Figure \ref{Figure-DataGraph-Filetering} (a) shows the CNIs computed for the data vertices. During degree and CNI computation, the label filter is applied and the vertices that do not verify this filter are pruned and are not considered in the degree and CNI of there neighbors. This is the case for vertices $v_{7}$, $v_{14}$ and $v_{15}$. only the data vertices that verify the label filter are considered when computing degrees and CNIs.

The first iteration, of the ILGF algorithm (see Figure \ref{Figure-DataGraph-Filetering} (b)), finds out that vertices $v_{1}$, $v_{3}$, $v_{5}$, $v_{13}$, $v_{16}$, $v_{17}$, $v_{19}$, $v_{20}$ and $v_{21}$ cannot be mapped to any query vertex because:
\begin{itemize}
\item  $v_{1}$, $v_{13}$, $v_{15}$, $v_{16}$, $v_{19}$, $v_{20}$ and $v_{21}$ do not pass the degree filter,
\item $v_{3}$ and $v_{5}$ do not pass the CNI filter.
\end{itemize}
After removing these vertices and updating the degree and CNI of their neighbors a new filtering iteration is triggered (see Figure \ref{Figure-DataGraph-Filetering} (c)). We note here that this second filtering iteration concerns only vertices $v_2$, $v_8$, $v_{11}$, and $v_{18}$ whose degrees and CNIs have been modified during the precedent filtering iteration. The second filtering iteration reveals that vertices $v_8$ and $v_{18}$ can also be pruned. They do not pass the degree filter. The resulting filtered data graph is depicted in Figure \ref{Figure-DataGraph-Filetering} (d) that also shows the new filtering opportunities triggered by the second filtering iteration. A third filtering iteration is launched on the vertices whose degrees and CNIs have been modified. Consequently, the third filtering iteration verifies only vertex $v_4$. $v_4$ does not verify the CNI filter and can be pruned leading to the graph depicted on Figure \ref{Figure-DataGraph-Filetering}(e). This figure shows also that the degree and CNI of vertices $v_2$ and $v_{10}$ are updated. Consequently these two vertices will be the target of the final filtering iteration that prunes vertex $v_2$ which no longer verifies the degree filter.  The final filtered graph is illustrated in Figure \ref{Figure-DataGraph-Filetering} (f).

\begin{figure*}
\centering
\includegraphics[scale=0.25]{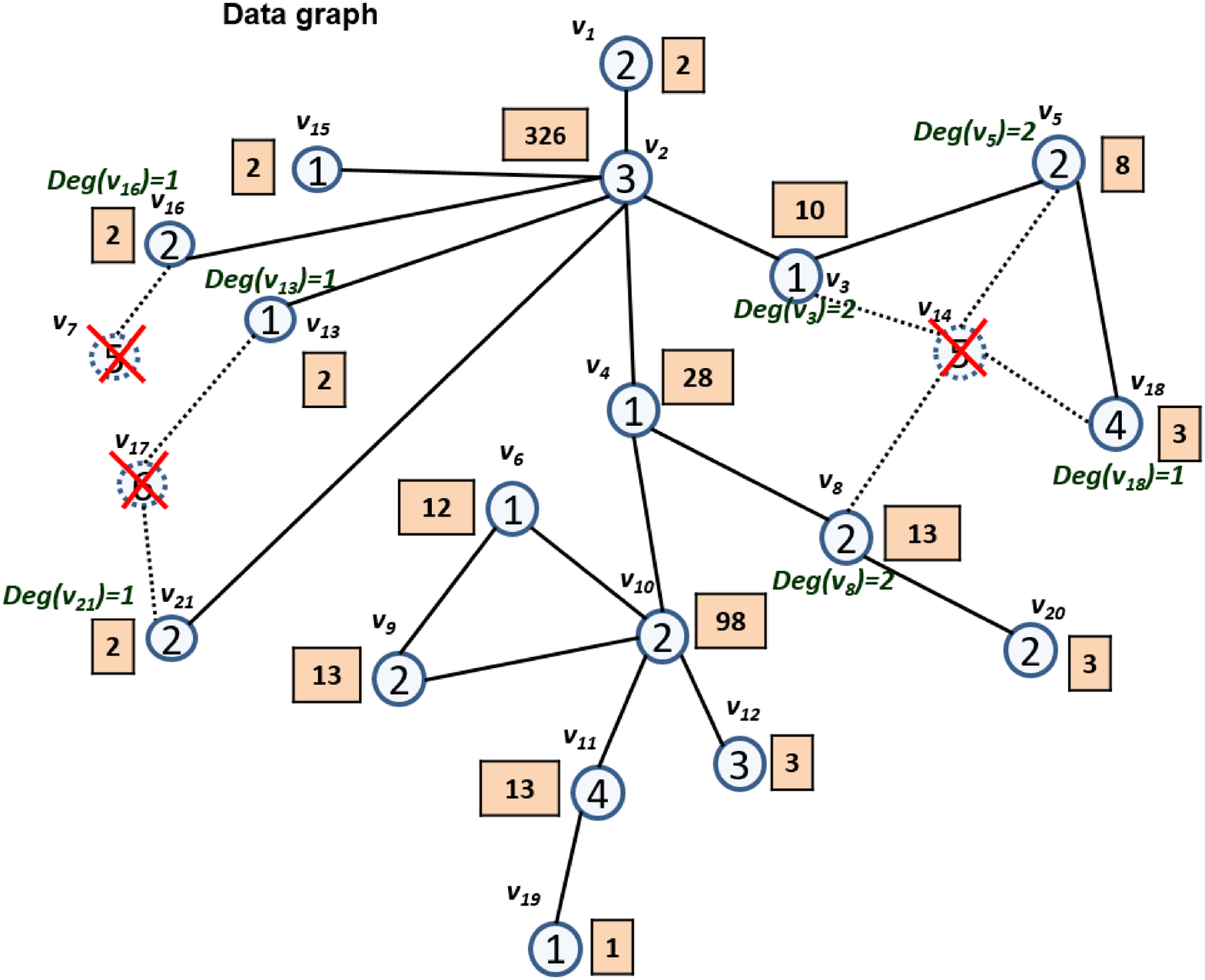}\quad
\includegraphics[scale=0.25]{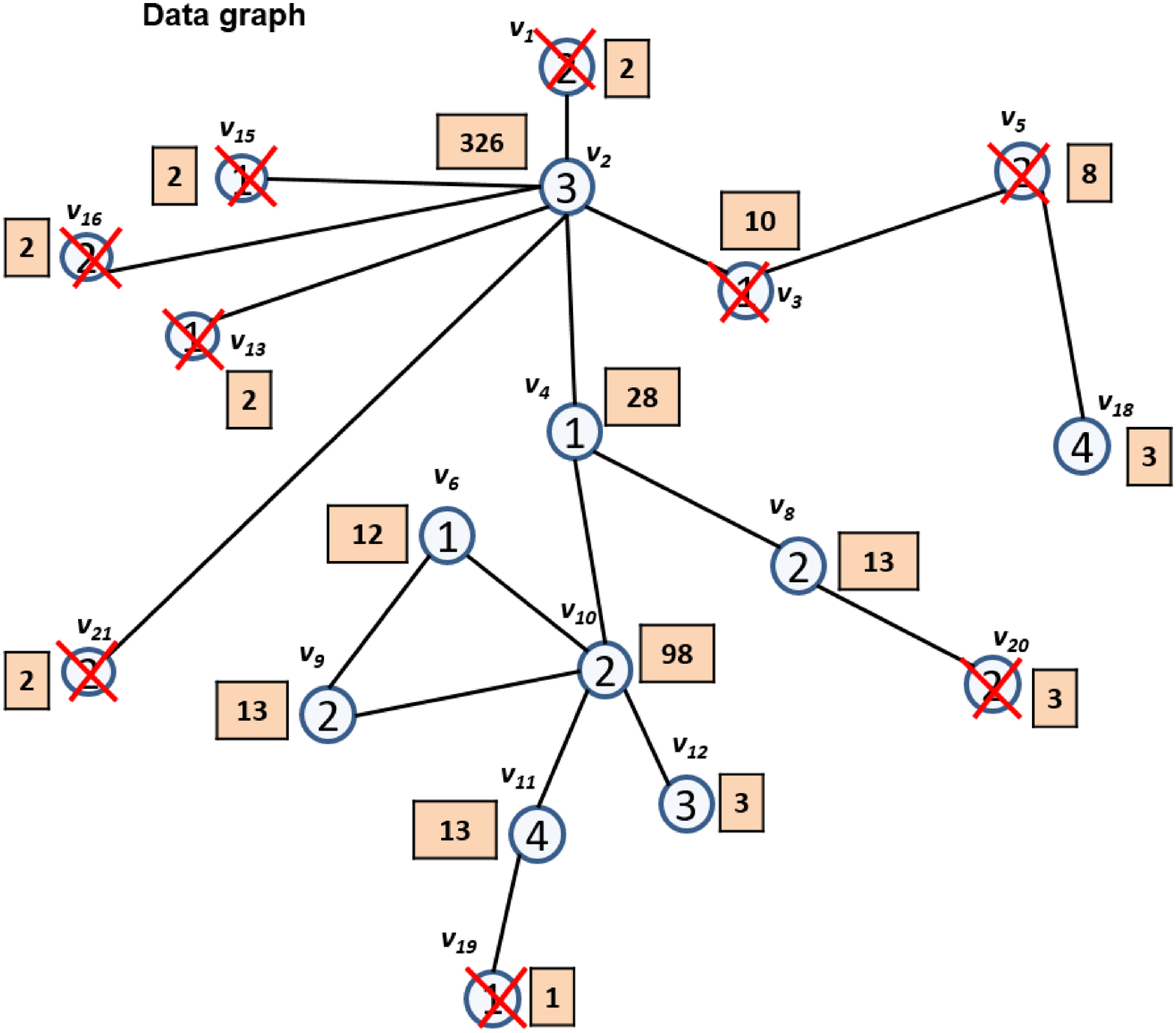}\\
(a) CNI computation and label filtering \hspace{4cm} (b) First filtering iteration\\
\includegraphics[scale=0.25]{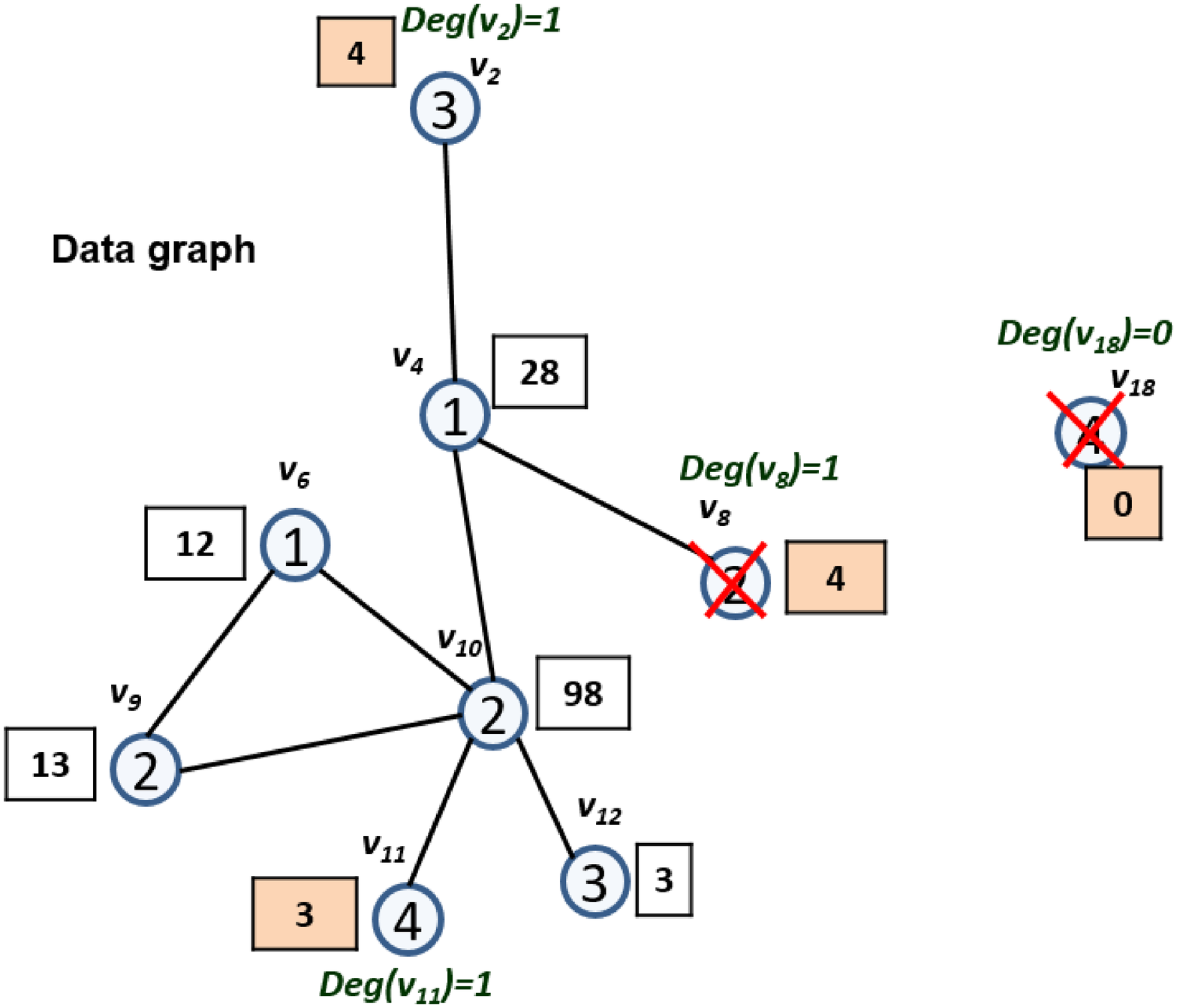}\quad \quad\quad \quad
\includegraphics[scale=0.25]{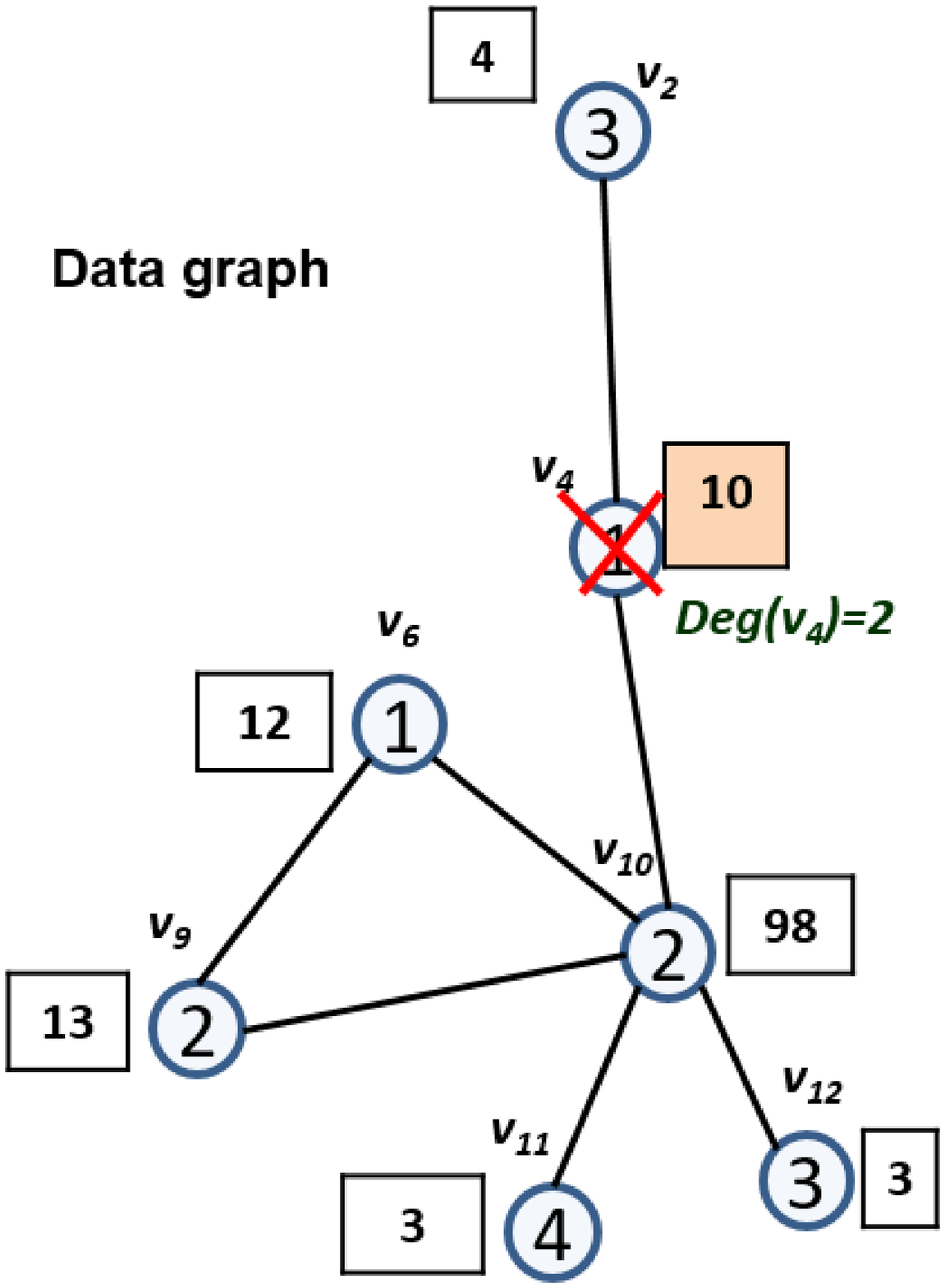}\\
(c) Second filtering iteration  \hspace{4cm} (d) Third filtering iteration\\
\includegraphics[scale=0.25]{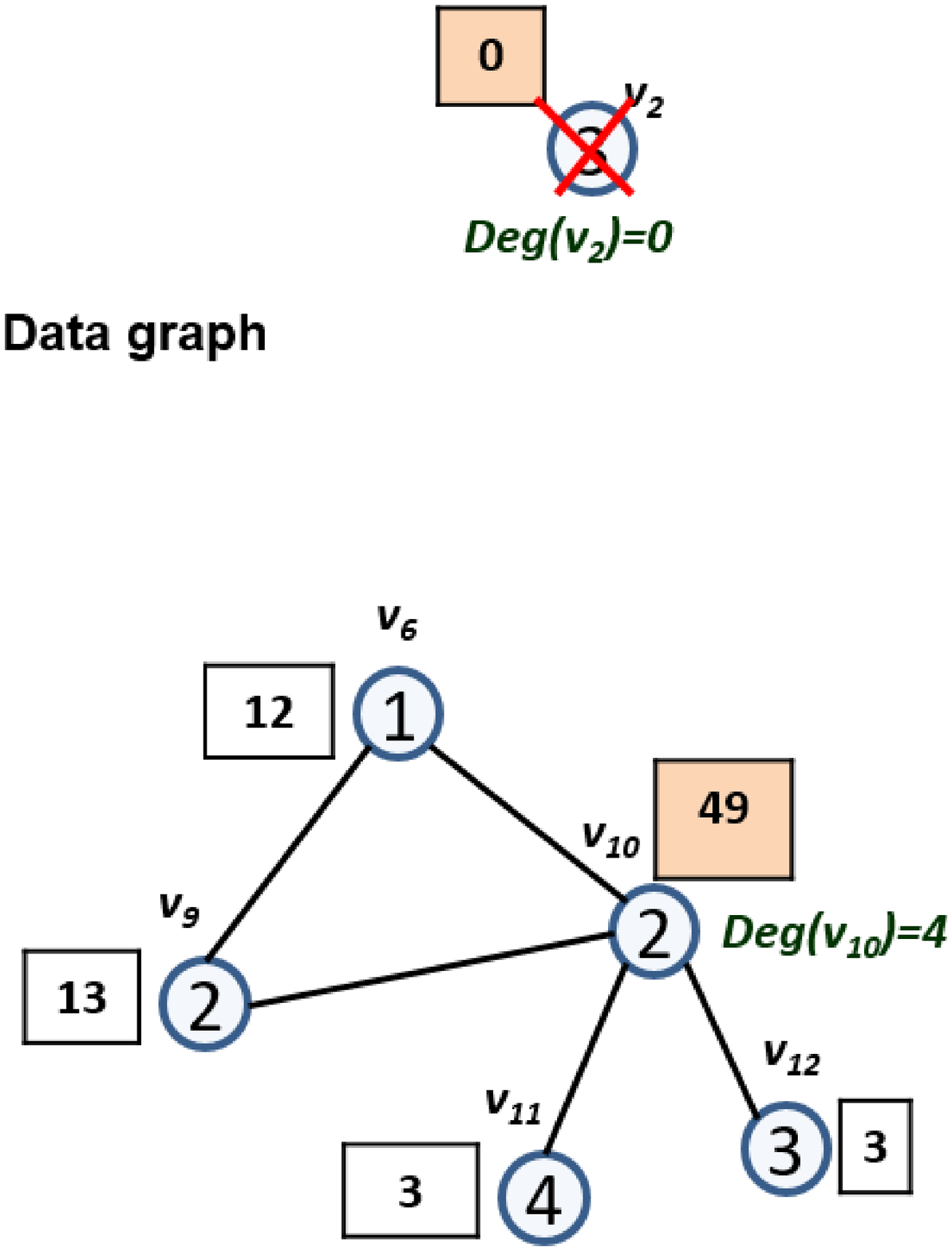}\quad \quad\quad \quad
\includegraphics[scale=0.25]{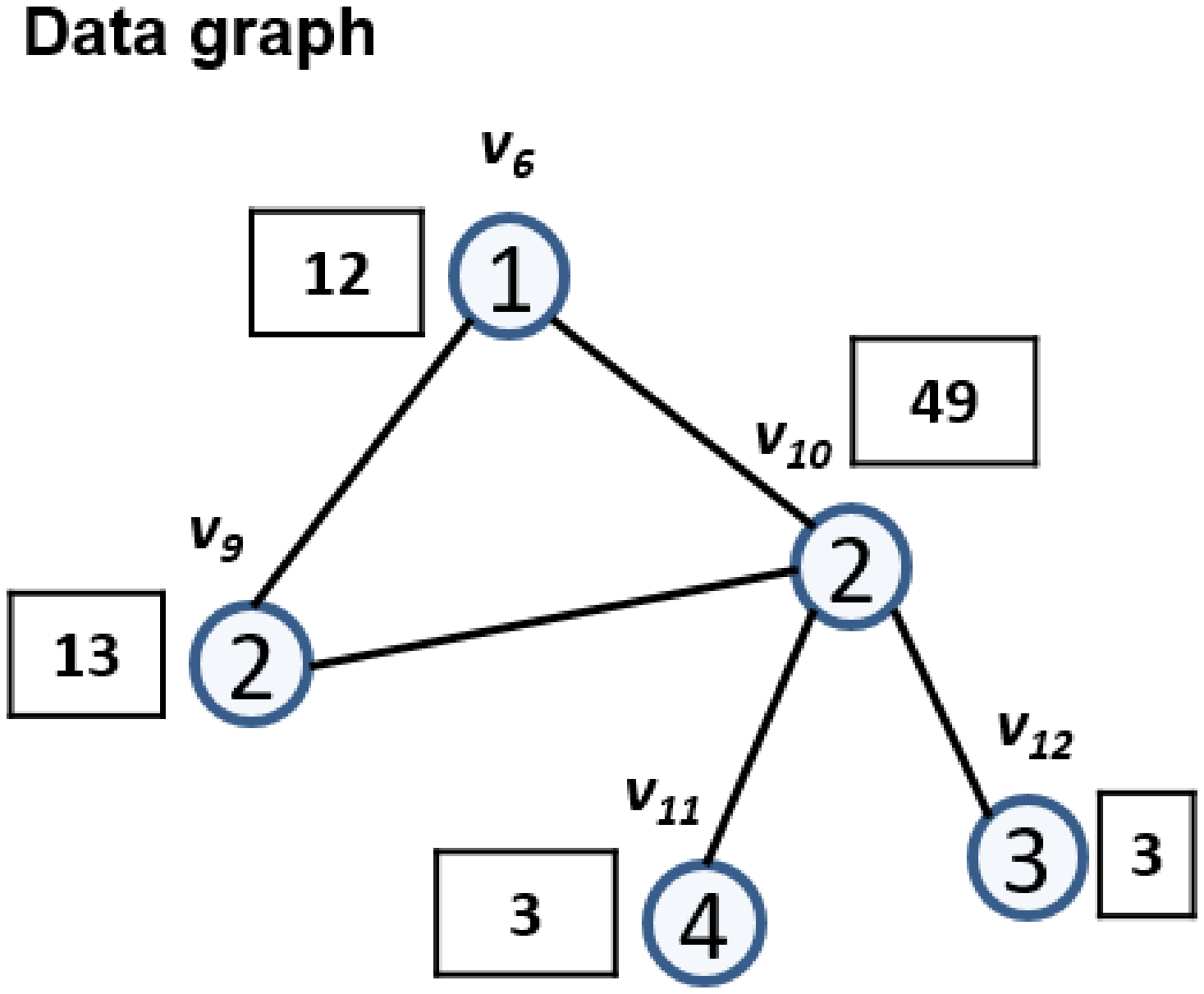}\\
(e)Fourth filtering operation\hspace{4cm} (f) Final filter data graph\\

\caption{Filtering iterations of our running example.\label{Figure-DataGraph-Filetering}}
\end{figure*}

\subsection{Subgraph Search}
After filtering, the data graph contains only the vertices that are candidates for query vertices, i.e., the vertices map at one-hop according to the CNI filter. Subgraph search allows to verify the mapping at k-hops. Algorithm \ref{Algo-SubgraphList} implements this step. It is a depth first search subroutine that parses the filtered data graph and lists the subgraphs of the filtered data graph that are isomorphic to the query by verifying the adjacency relationships. This step allows also to handle edge labels by discarding those that do not match the query labels. The subroutine \textit{neighborCheck()} verifies that a mapping $(v, u)$ is added to the current partial embedding $M$ only if $v$ and $u$ have neighbors that also map.

\begin{algorithm}
\KwData{a partial embedding $M$.}
\KwResult{All embeddings of $Q$ in $G$.}
\Begin{
\If {$|M| = |V(Q)|$}{
Report $M$\;
}
Choose a non matched vertex $u$ from $V(Q)$\;
$C(u) \leftarrow \{$ non matched $ v \in V(G)$ such that cniVerify($v, u)$)\;
\ForEach {$v \in C(u)$}{
  \If {neighborCheck($u$,$v$, $M$)} {
     $M \leftarrow M\cup \{(u,v)\}$\;
     SubgraphSearch($M$)\;
     Remove $(u,v)$ from $M$ \;
     }
}
}
\caption{SubgraphSearch.\label{Algo-SubgraphList}}
\end{algorithm}

\begin{algorithm}
\KwData{a partial embedding $M$, a query vertex $u$ and a data vertex $v$.}
\KwResult{returns true if $u$ and $v$ have neighbors that match.}
\Begin{
\Return $\forall (u', v')\in M, ((u, u') \in E(Q) \rightarrow (v, v') \in E(G) \bigwedge \ell((u, u') )=\ell((v, v'))$
}
\caption{Function\textit{ neighborCheck}$(u,v,M)$.\label{Algo-neighborCheck}}
\end{algorithm}

\subsection{Extension to Directed/Edge Labeled Graphs}
A compact neighborhood index can also be computed for edge labeled graphs as well as for directed graphs. To take into account edge labels, a simple solution is to compute separately, a CNI, denoted $cni_v(u)$, for the vertices and a CNI, denoted $cni_e(u)$, for the edges and then given those two values, compute a global CNI of the considered vertex. $cni_v(u)$ and $cni_e(u)$ are computed using the bijection defined in Section \ref{sec-cni}. If edges and vertices are defined on the same set of labels, $cni_v(u)$ and $cni_e(u)$ can be computed  using the same dimension, i.e., the same value for $k$. If edges and vertices are defined on different sets of labels, $cni_v(u)$ and $cni_e(u)$ are computed with different dimensions. The global CNI of the vertex is computed with $k=2$.

An example is given in Figure \ref{Figure-EXTCNI}(a). In this example, vertices and edges use the same set of labels $\{a, b, c, d\}$ encoded by the set of integers $\{1, 2, 3, 4\}$. For instance, for vertex $u_1$,  we have:
 \begin{itemize}
 \item $cni_v(u_1)=g_4(0,2,0,0)=12$, and
 \item $cni_e(u_1)=g_4(0,1,0,1)=7$.
 \end{itemize}
 Then $cni(u_1)$ is computed using a bijection of dimension $2$ as follows:\\
 $cni(u_1)=g_2(cni_v(u_1), cni_e(u_1))=g_2(12,7)=202$.\\

 For directed graphs, we use a similar approach. We compute a CNI for vertex labels and two CNIs for edge labels: a CNI for ingoing edges, $cni_{IN}(v)$,  and a CNI for outgoing edges, $cni_{OUT}(v)$,. A global CNI for the vertex is obtained using these three values and the bijection with dimension 3 as depicted in Figure \ref{Figure-EXTCNI}(b):
 $cni(u_1)=g_3(cni_v(u_1), cni_{IN}(u_1), cni_{OUT}(u_1))=g_3(12,1,3)=919$.\\

\begin{figure*}
\centering
\includegraphics[scale=0.30]{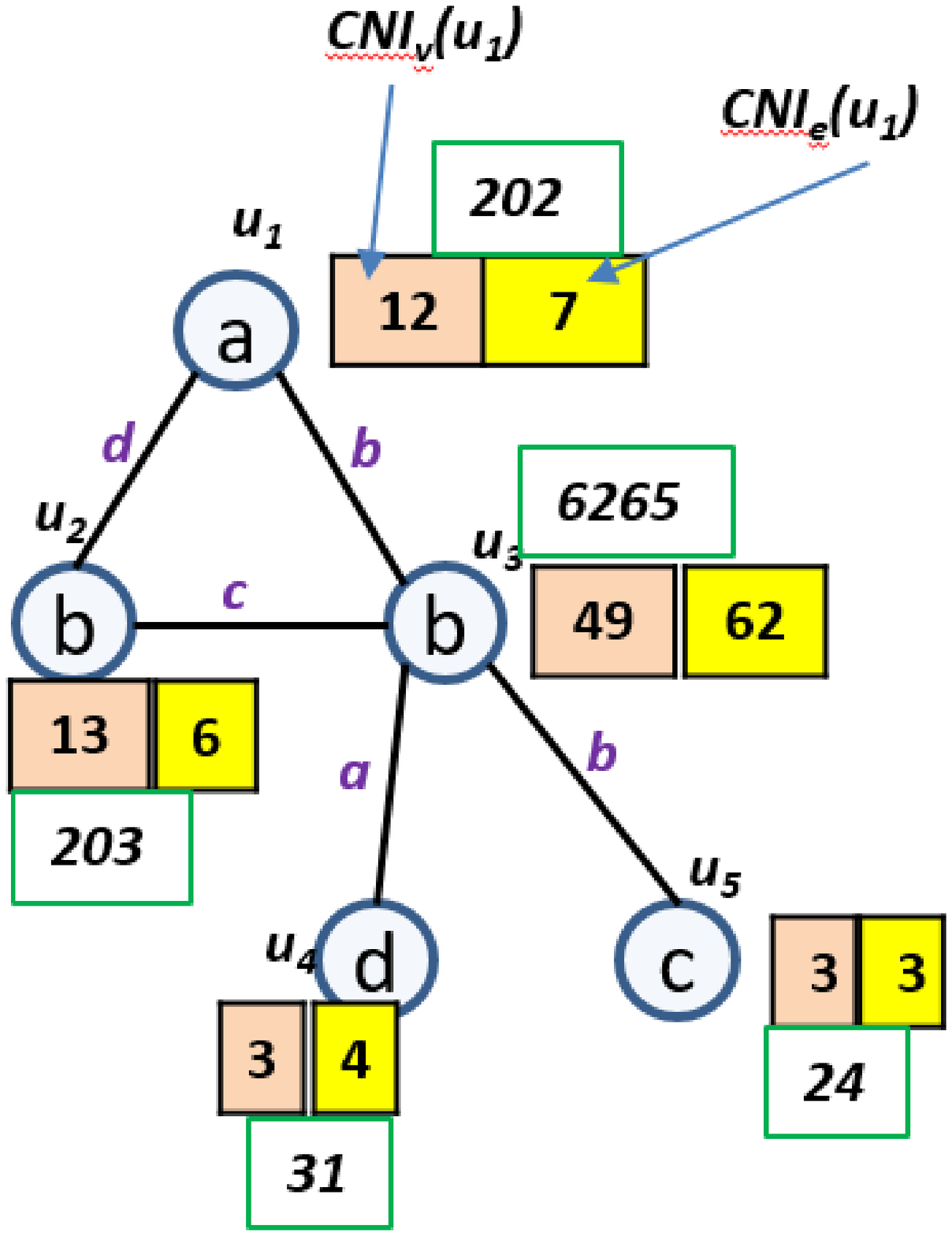}\quad
\includegraphics[scale=0.30]{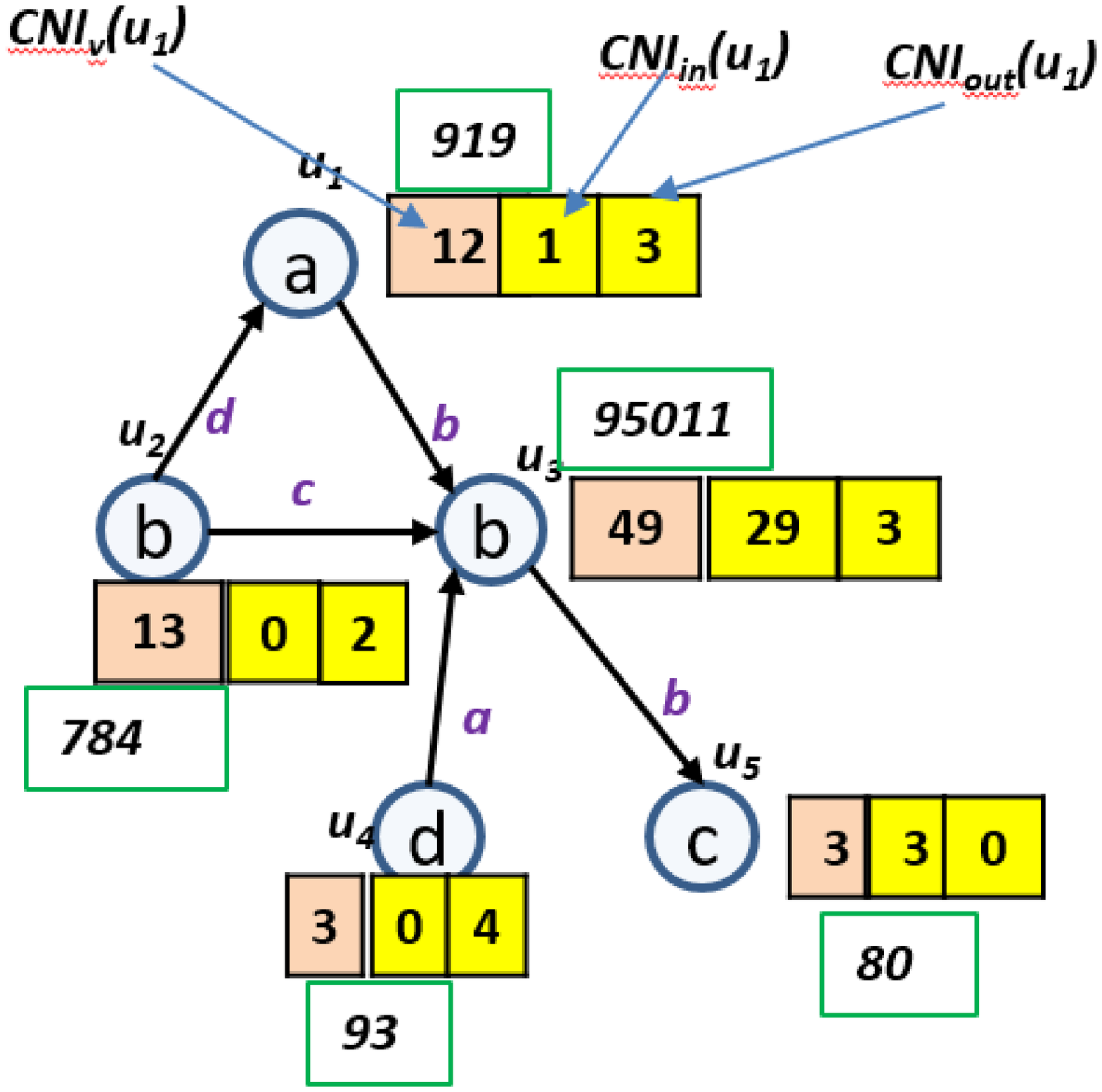}\\
(a) Edge Labeled Graph \hspace{4cm} (b) Directed Graph\\
\caption{CNIs for Directed and Edge Labeled Graphs.\label{Figure-EXTCNI}}
\end{figure*}
\section{$cni(v)$ at $(q\geq1)$-hops Neighborhood}
\label{app-proofcnik}
The compact neighborhood index can also be computed for the $q$-hops neighborhood with $q\geq1$. According to the desired power of the obtained filter, two algorithms can be considered:
\begin{enumerate}
\item A light filter featuring the whole $q$-hops neighborhood  with one single  CNI, denoted $cni_{1-q}(v)$. In this case, we consider, in the CNI formula given in Section \ref{sec-cni}, all the vertices reachable from vertex $v$ within range $q$, i.e., the ball of radius $q$ centred at vertex $v$. For instance, the CNI of vertex $u_1$ of the query of our running example (depicted in Figure \ref{Fig-Example}) featuring its neighborhoods at $2$-hops includes vertices $u_2$, $u_3$, $u_4$ and $u_5$ and is given by:\\
    $cni_{1-2}(u_1)=g_4(0,2,1,1)=48$. \\
    A data vertex $v$ is a candidate for a query vertex $u$ if $cni_{1-q}(v)\geq cni_{1-q}(u)$. However, two vertices having the same number of neighbors in their $q$-hops neighborhood and the same CNI at $q$-hops may not be isomorphic.
\item A strong filter featuring each $q$-hops neighborhood  with its own CNI. In this case, the CNI of vertex $v$ featuring its neighborhood at $q$-hops is a set of $q$ integers $\{cni_1(v), cni_2(v),\cdots ,cni_q(v) \}$ where $cni_j(v)$ is the CNI computed considering the neighbors of $v$ at exactly $j$-hops in a shortest path from $v$. $cni_1(v)$ is the CNI computed in Section \ref{sec-cni}. The other value are computed with the same formula and on the same set of labels as follows:
$cni_j(v)=\sum\mbox{\ensuremath{_{i=1}^{k}\hbar(i,x_1+...+x_i)}}$  where $k$ is the number of distinct labels in the query, and $x_i$, is the  number of occurrences of the label indexed $i$ in the $j$-hops neighborhood of $v$.  For instance, the CNI at $j=2$ of the query vertex $u_1$ of our running example (see Figure \ref{Fig-Example}) comprises vertices $u_4$ and $u_5$ and can be computed similarly to the CNI of its direct neighborhood, by:
$cni_2(u_1)= g_4(0,0,1,1)=6$.
The CNI  at $q$-hops can be used to prune the data vertices that are not candidate for a query vertex but that passe through the $(q-1)$-hops CNI as follows:

\begin{lem} [$q$-hop Degree filter] \label{lem-dk}
Given a query $Q$ and a data graph $G$, a data vertex $v \in V(G)$ is not a candidate of $u \in V(Q)$ if $deg_{\mathcal{L}(Q)}^q((v)<deg_{\mathcal{L}(Q)}^q(u)$ where $deg_{\mathcal{L}(Q)}^q(u)$ is the number of vertices reachable from $u$ with exactly $q$-hops in a shortest path from $u$ and have a label in $\mathcal{L}(Q)$.
\end{lem}

\begin{lem} [$CNI_{q}$ filter] \label{lem-CNIk}
Given a query $Q$ and a data graph $G$, a data vertex $v \in V(G)$ that verifies the $CNI_{q-1}$ filter and the $(q)$-hops Degree filter is not a candidate of $u \in V(Q)$ if $cni_{q}(v)<cni_{q}(u)$.
\end{lem}
Lemma \ref{lem-dk} is straightforward. The proof of Lemma \ref{lem-CNIk} is similar to the proof of Lemma \ref{lem-cni}.

We note also that the strong filter can also be represented by a single CNI by computing it on the set $\{cni_1(v), cni_2(v),\cdots ,cni_q(v) \}$, i.e., $g_k(cni_1(v), cni_2(v),\cdots ,cni_q(v))$.
\end{enumerate}

\subsection{Extension to Larger Graphs}

For large data graphs, we aim to keep in memory as few vertices and edges as the three filters can achieve. So, filtering begins while reading the data graph. For this, we compute vertex degrees and CNIs incrementally during graph parsing. Only a single pass of the graph is needed. This is important if we deal with a graph stream or a sequential read of a graph from disk, i.e., a graph that does not fit into main memory and that is loaded  part by part. We keep in memory only the vertices (and the corresponding edges) that verify the label, degree and CNI filters. These are the vertices and edges that will be used during subgraph search.
As we parse the data graph, the label filter is straightforward. However, the degree and the CNI can be used when their values, computed incrementally, are sufficient for pruning. However, this depends on how the stream of edges arrives. If edges are sorted, i.e., we access all the edges involving vertex $i$, then all the edges involving vertex $i+1$ and so on,  the amount of pruning will be larger during the parse than in the case edges arrive randomly.

Algorithm \ref{Algo-DataProcessing} presents the filtering actions performed during the data graph reading in the case where edges are sorted. In this case, the three filters can be applied as the edges of  a vertex are accessed avoiding to store them. When all the edges incident to the current vertex are available (see lines 10-17), we can compute the CNI of the current data vertex and compare it with the CNIs of the query vertices (see lines 19-23). If the vertex does not verify the CNI filter, the vertex and all its edges are pruned (see lines 20-23). The filtered data graph, denoted $G_Q$ obtained at the end of the reading-filtering process is then handled by the ILGF algorithm for global filtering.

\begin{algorithm}
\KwData{A Data Graph $G$ (stream of edges). }
\KwResult{A filtered data graph $G_Q$.}
\Begin{
//processing a stream of sorted edges \\
$V(G_Q) \leftarrow \emptyset$\;
$E(G_Q) \leftarrow \emptyset$\;
\Repeat  {end of stream}{
    read edge $(x,y)$\;
   \If {($x \notin V(G_Q)$ and $\ell(x) \in \mathcal{L}(Q)$)}{
    $V(G_Q) \leftarrow V(G_Q) \cup \{x\}$\;
    $current \leftarrow x$\;
    \While {$x=current$}{
    \If {($\ell(y) \in \mathcal{L}(Q)$)}{
      \If{($(y \notin V(G_Q))$)}{
      $V(G_Q) \leftarrow V(G_Q) \cup \{y\}$\;
      $E(G_Q) \leftarrow E(G_Q) \cup \{(x,y)\}$\;
      }
    }
    read edge $(x,y)$\;
  }
  compute cni(current)\;
  \If {($\forall u \in V(Q), !cniVerify(current,u)$)}{
     remove current from $V(G_Q)$\;
     remove all the edges of $current$ from $E(G_Q)$\;
         }
         }
 }
ILGF($G_Q$)\;
}
\caption{Large Data Graph Filtering .\label{Algo-DataProcessing}}
\end{algorithm}
\subsection{Discussion}
As presented in Section \ref{sec-cni}, the generalized “Cantor $k$-tupling function” used to compute the CNIs is a polynomial of degree $k=|\mathcal{L}(Q)|$. Its formula needs to compute factorials. However, factorials increase faster than all polynomials and exponential functions \footnote{https://en.wikipedia.org/wiki/Factorial}. Even if the CNI formula does not grow as fast as a factorial because it is a fraction with factorials in the bottom that will decrease its growth, it is not feasible to compute factorials using the classical programming types: with a C++ \textit{unsigned long long int}, which requires $8$ bytes on a windows machine, the maximum integer for which we can compute a factorial is $65$ and it gives the $19$ digits integer $9223372036854775808$. Even if $p!(s-1)!$ are dividing quantities that will decrease the value of the CNI, we must compute them to get the final value of the CNI even if it is much smaller than the value of the computed factorials.

To deal with this, we first reduce the CNI formula to its simplest representation. We recall that a CNI is computed by the following formula:
$$g_{k}(x_1,x_2,x_3,\cdots,x_k)=\sum\mbox{\ensuremath{_{j=1}^{k}\hbar(j,x_1+...+x_j)}}$$
  and $$\hbar(p,s)=\binom{s+p-1}{p}=\frac{(s+p-1)!}{p!(s-1)!}$$
  However, $\hbar(p,s)$ can be simplified. In fact, $\frac{(s+p-1)!}{p!(s-1)!}=\frac{(s+p-1)(s+p-2)\cdots(s+1)s}{p!}$. So, the top of the fraction is a product of $p$ terms. Furthermore, we can reduce the size of this product without computing $p!$. Because its final result is an integer, the top of the fraction must be a multiple of its bottom, i.e., $p!$. Accordingly, we do in priority the divisions as detailed by Algorithm \ref{algo-reduction}. As soon as the top $st$ of the fraction is great enough to be divided by a term from $p!$, the division is performed (see lines 10-15). This allows us to compute CNIs without unnecessary out-of range errors due to factorial computations.

\begin{algorithm}
\KwData{integers $s$ and $p$. }
\KwResult{$\hbar(p,s)=\frac{(s+p-1)(s+p-2)\cdots(s+1)s}{p!}$.}
\Begin{
\If {$p>(s+p-1)$} {\Return 0\;}
\If {$(s+p-1=1)$} {\Return 1\;}
$st \leftarrow 1$\;
$i \leftarrow 1$\;
\For{($t = s;\ t\leq(s+p-1);\ t=t+1$)}{					
	$st \leftarrow st*t$\;
\While {$(i\leq p) \wedge (st \% i =0) \wedge (st>0)$}{
			$st \leftarrow st/i$\;
			$i \leftarrow i+1$\;
	}

	}
\Return $st$\;
}
\caption{Optimising the Computation of CNIs .\label{algo-reduction}}
\end{algorithm}

It is also worth noting that in the formula of CNI, $k=|\mathcal{L}(Q)|$ and  $\sum\mbox{\ensuremath{_{i=1}^{k} x_i}}\leq \Delta$ where $\Delta$ is the maximum degree in the data graph $G$. $k$ is bounded by $|V(Q)|$ but  $\Delta$  depends on the application and the type of graphs that it deals with. The formula of the CNI grows similarly to $\frac{\Delta^k}{k!}$. To adapt to big values of $k$ and $\Delta$, we worked on two main solutions:
\begin{enumerate}
\item A divide an conquer approach parameterised by the values of $k$ and $\Delta$. The main idea here is to process the $k$ labels in small packets of $s$ labels each. The value of $s$ is chosen so that to ensure big values for the $x_i$. In this case, we compute a CNI for each packet of labels. However this means that comparing two vertices is no more feasible in constant time, i.e., $\mathcal{O}(1)$, it is $\mathcal{O}(k/s)$.
\item A specialised library for large numbers, which is a more general solution. In our experiments, we used the GMP library \footnote{https://gmplib.org/} which is designed to process large numbers independently of the size of programming language types. With this library, one can handle any large number. The sole limitation is the available main memory. It is also designed to be as fast as possible, both for small operands and for huge operands for all arithmetic operations.  Another interesting feature of this library is that memory for storing the number is allocated incrementally with the growth of the number. When, the number is no more needed, we can free the space allocated to its storage. In our implementation, we do not keep  large CNIs, once computed, we compute and store the logarithm of the CNI and we free the space allocated for the CNI. We recall that logarithm is a bijective function that grows very slowly for large numbers and are consequently used to compress large-scale scientific data and store it. So, obtaining large values of CNIs have no negative effect on the proposed approach.
\end{enumerate}
\section{Experiments}
\label{Sec-Evaluation}
We evaluate the  performance of our algorithm, \textit{CNI-Match} (for Compact Neighborhood Index based Matching), over various types of graphs, sizes of queries and  number of labels. We also compare it with the most recent state of the art algorithm, CFL-Match \cite{Bi2016}. Note that CFL-Match is compared to the other existing solutions, such as Turbo$_{ISO}$, QuickSI and SPath, and showed to be more efficient in \cite{Han2013,Ren2015,Bi2016}.

For a fair comparison, we implemented the two algorithms\footnote{The source code of the two algorithms is available on Git and can be easily provided.} in the same environment and framework using C++ and the SNAP library\footnote{Stanford Network Analysis Platform. http://snap.stanford.edu/}. We also used compiling option $-O3$. For CNI-Match, we used the GMP\footnote{https://gmplib.org/} specialised library to compute factorials and store them. All experiments are performed on an \textit{Intel} $i5$ $3.50$ GHz, $64$ bits computer with $8$ GB of RAM running windows 7.

We first describe the datasets used in the experiments, then we present our results.
\subsection{Datasets}
\noindent We use three main datasets which are known datasets used by almost all existing methods in their evaluation process. So, we mainly use them as comparative datasets. The underlying graphs represent protein interaction networks coming from three main organisms: human (HUMAN and HPRD datasets) and  yeast (YEAST dataset). The Human dataset consists of one large graph representing a protein interaction network. This graph has $4,675$ vertices and $86,282$ edges. It is a dense graph with an average degree of $36.9$. The number of unique labels  is $44$. The HUMAN dataset is a known difficult instance for subgraph isomorphism search \cite{Katsarou2017}. HPRD is a graph that contains $37,081$ edges and $9,460$ vertices. The number of unique labels in the dataset is $307$ and the average degree is  $7.8$. YEAST contains $12,519$ edges, $3,112$ vertices, and $71$ distinct labels. The average degree of the graph is $8.1$. The HUMAN dataset is available in the RI database of biochemical data \cite{Bonnici2013}. HPRD and YEAST come from the work of \cite{Lee2013} and \cite{Bi2016}.
\begin{table} [t]
	\caption{Graph Dataset Characteristics. }
$|\Sigma|$ is the number od distinct labels in the data graph.\\
	\label{Table-datasets}
	\centering
	\begin{scriptsize}
	\begin{tabular}{|c|c|c|c|c|}
		\hline
		Dataset     &  $|V|$    & $|E|$        & $|\Sigma|$ \tabularnewline
		\hline
		HUMAN           & 4,675       & 86,282          &44           \tabularnewline
		HPRD            &9,460       &37,081           &307   \tabularnewline
		YEAST           &3,112       &12,519           &71           \tabularnewline
\hline		
	\end{tabular}
			\end{scriptsize}
	
\end{table}
\begin{figure*}[t]
\centering
\includegraphics[ scale=0.30]{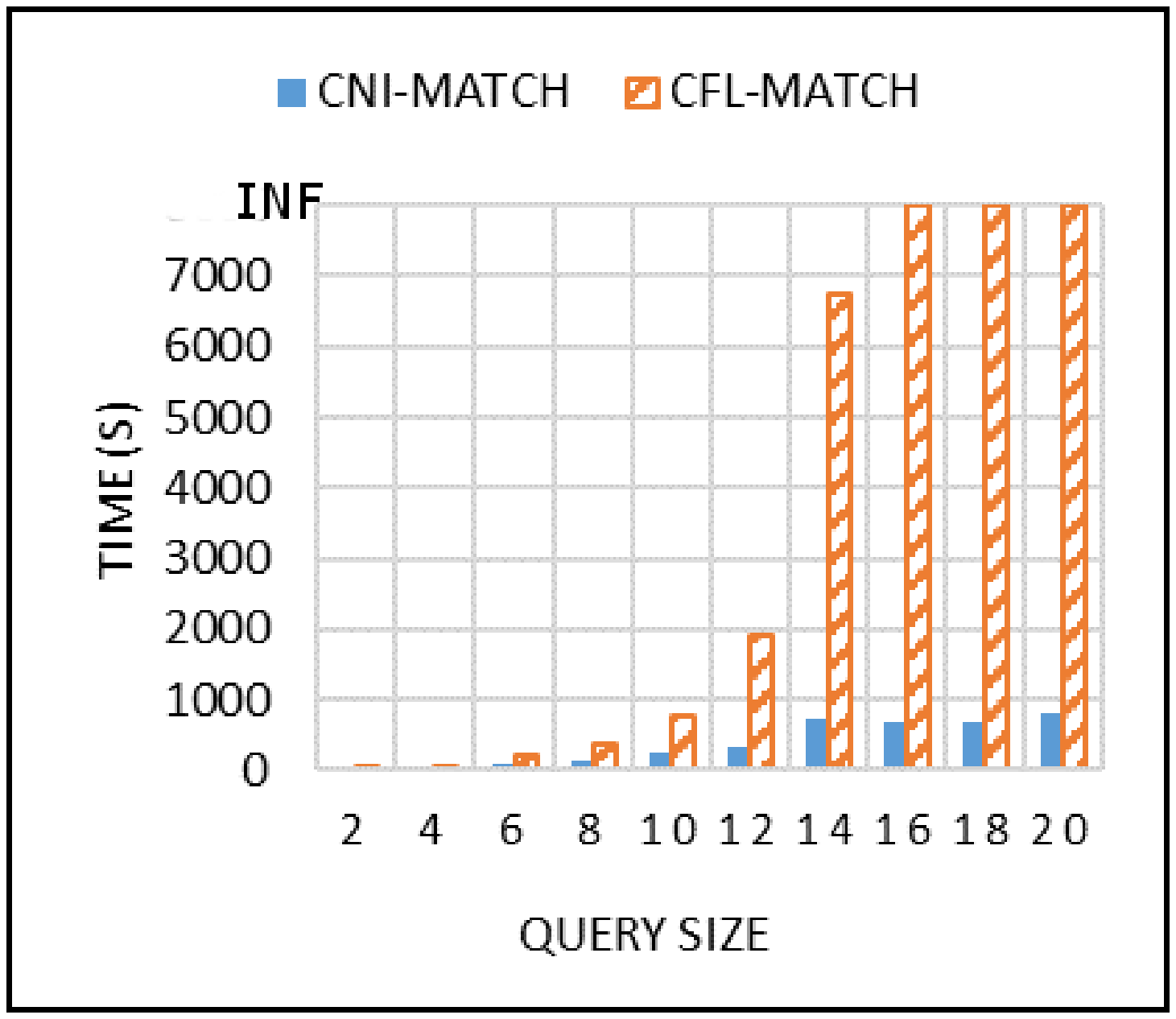}
\hspace{1cm}
\includegraphics[ scale=0.30]{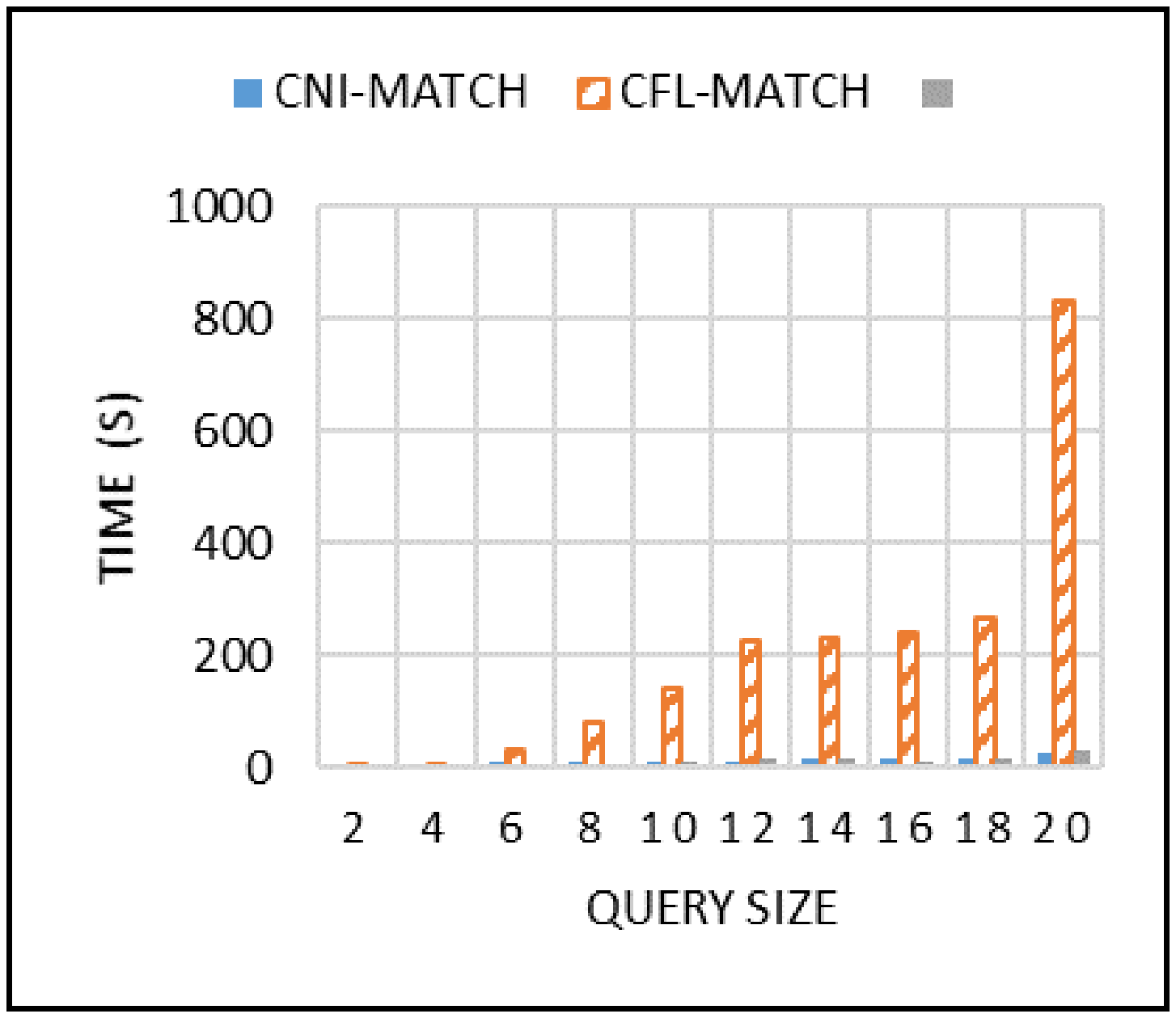}
\hspace{1cm}
\includegraphics[ scale=0.30]{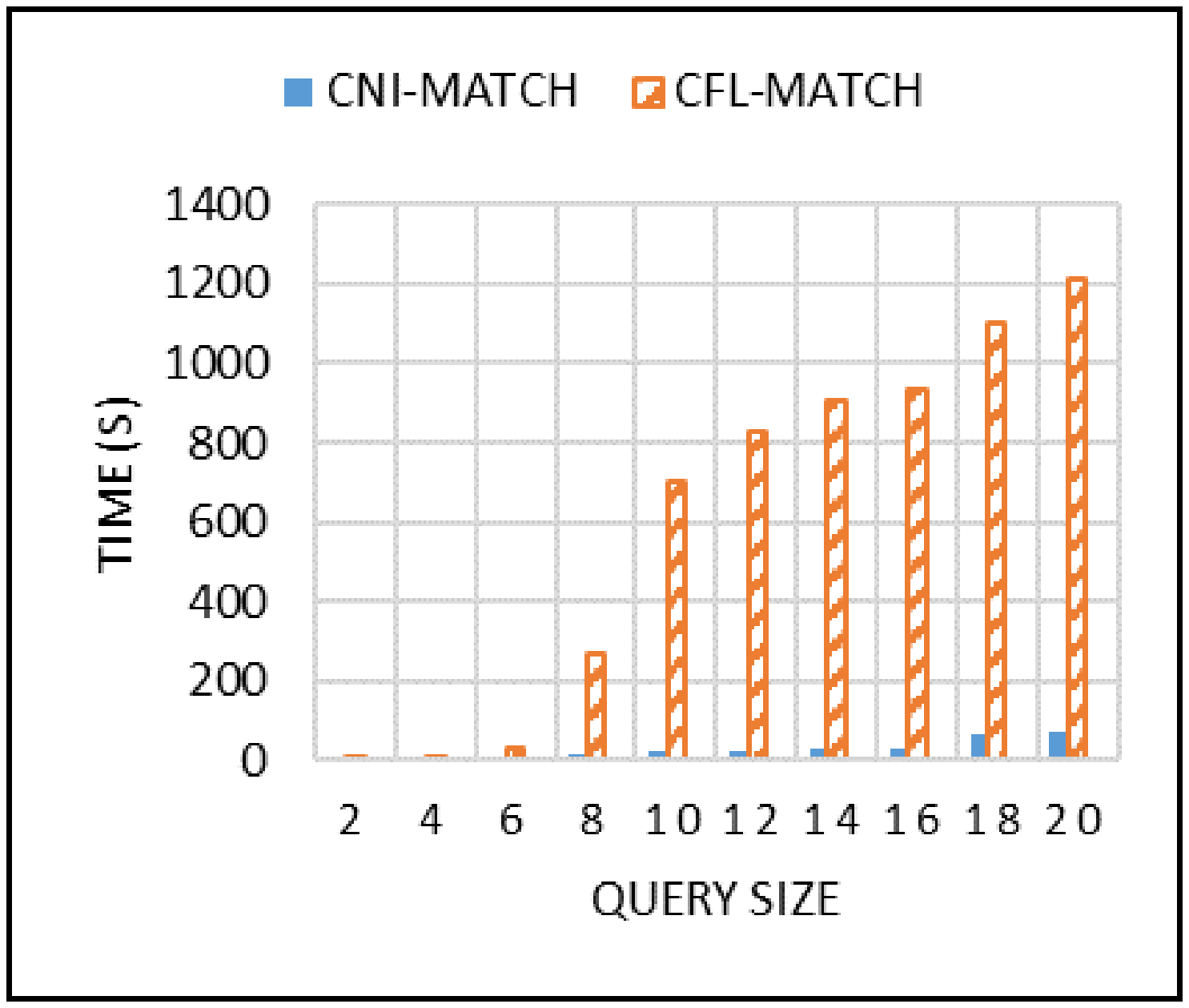}\\
(a) HUMAN dataset    \hspace{2cm}(b) YEAST dataset   \hspace{2cm}  (c) HPRD dataset
\caption{Time performance on sparse queries (varying $|V(Q)|$). \label{Fig-SmallDS}}
\end{figure*}
\begin{figure*}[t]
\centering
\includegraphics[ scale=0.30]{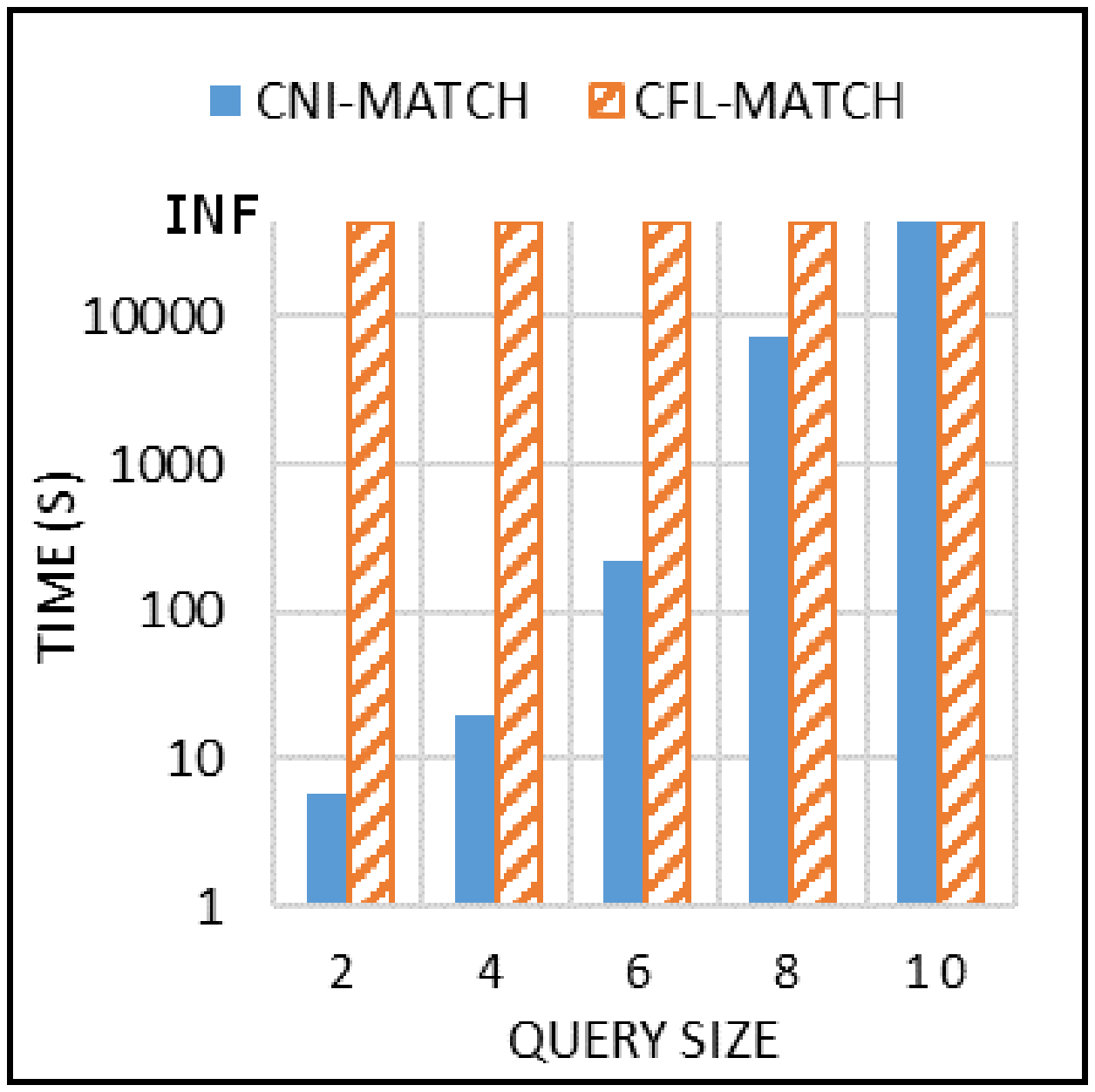}
\hspace{1cm}
\includegraphics[ scale=0.30]{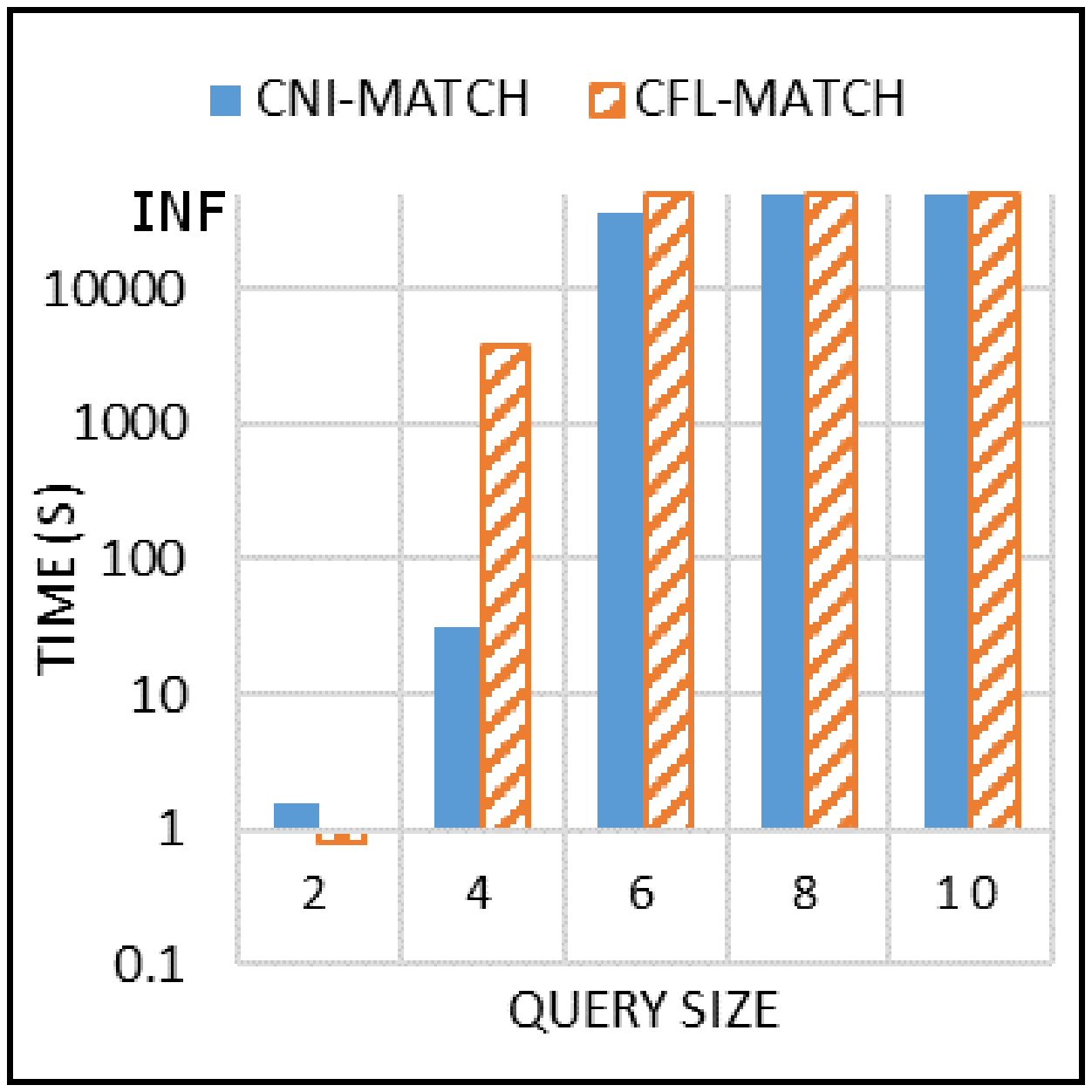}
\hspace{1cm}
\includegraphics[ scale=0.30]{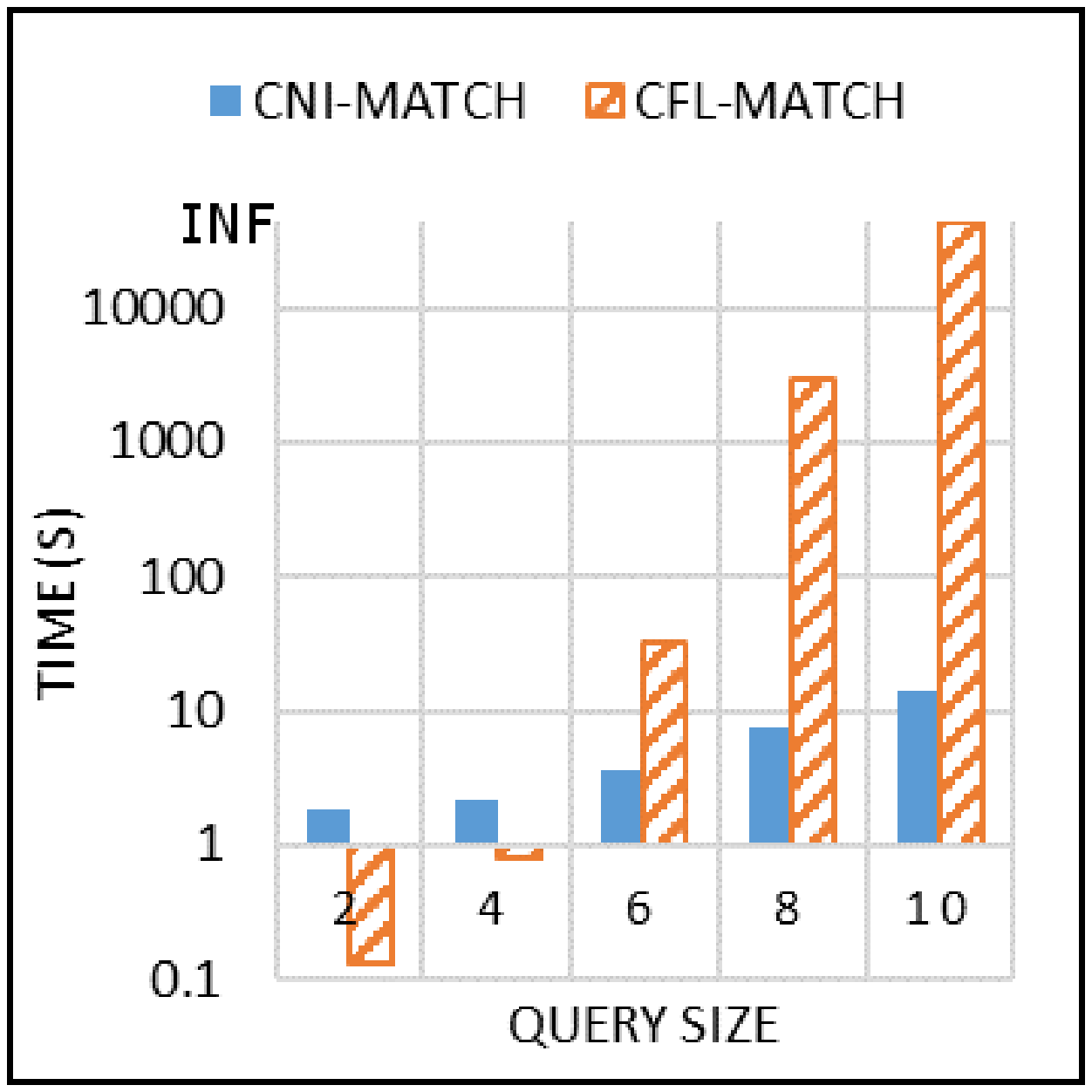}\\
(a) HUMAN dataset    \hspace{3cm}(b) YEAST dataset   \hspace{2cm}  (c) HPRD dataset (logscale)
\caption{Time performance on dense queries (varying $|V(Q)|$). \label{Fig-SmallDS2}}
(Results are in log scale)
\end{figure*}
\begin{figure*}[t]
\centering
\includegraphics[ scale=0.30]{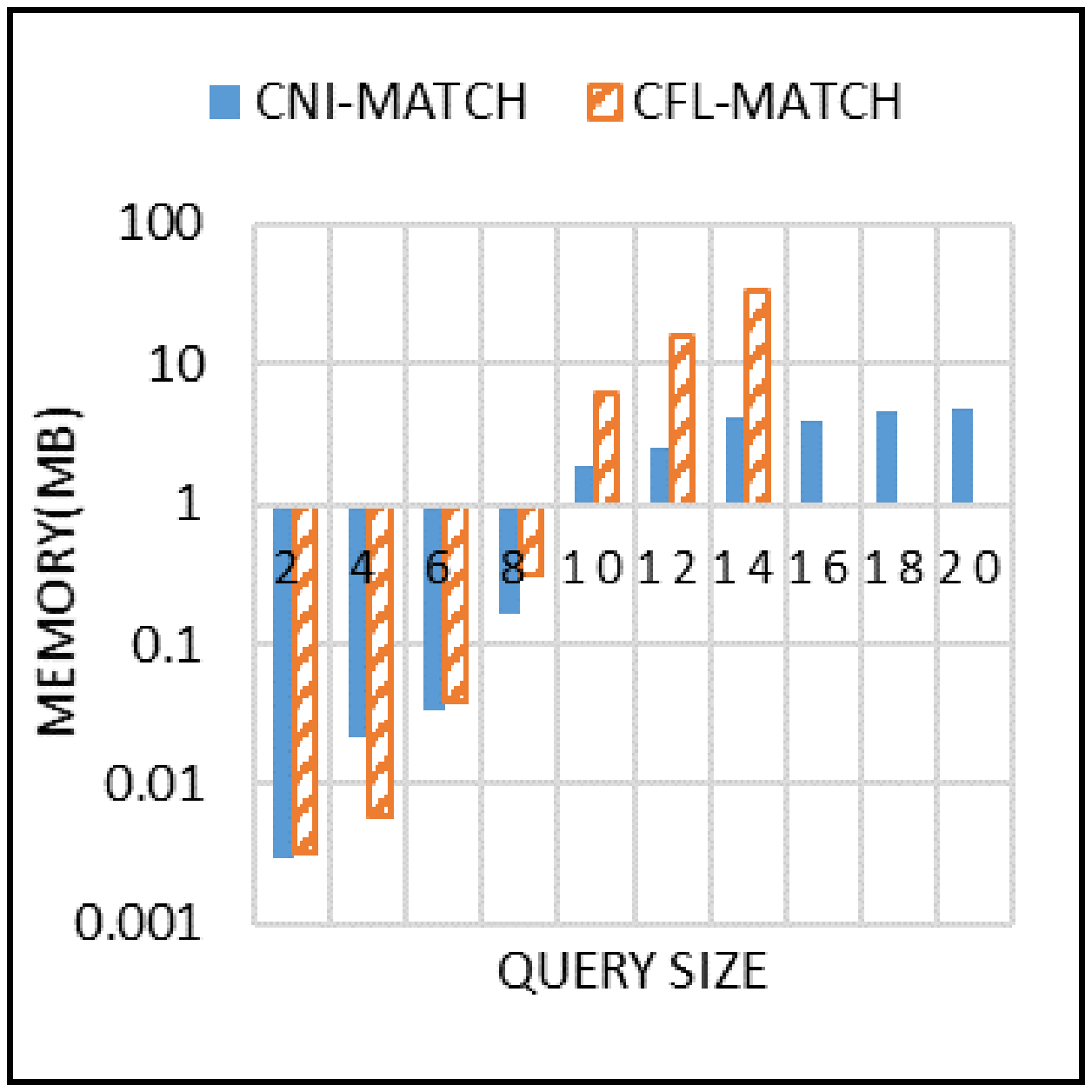}
\hspace{1cm}
\includegraphics[ scale=0.30]{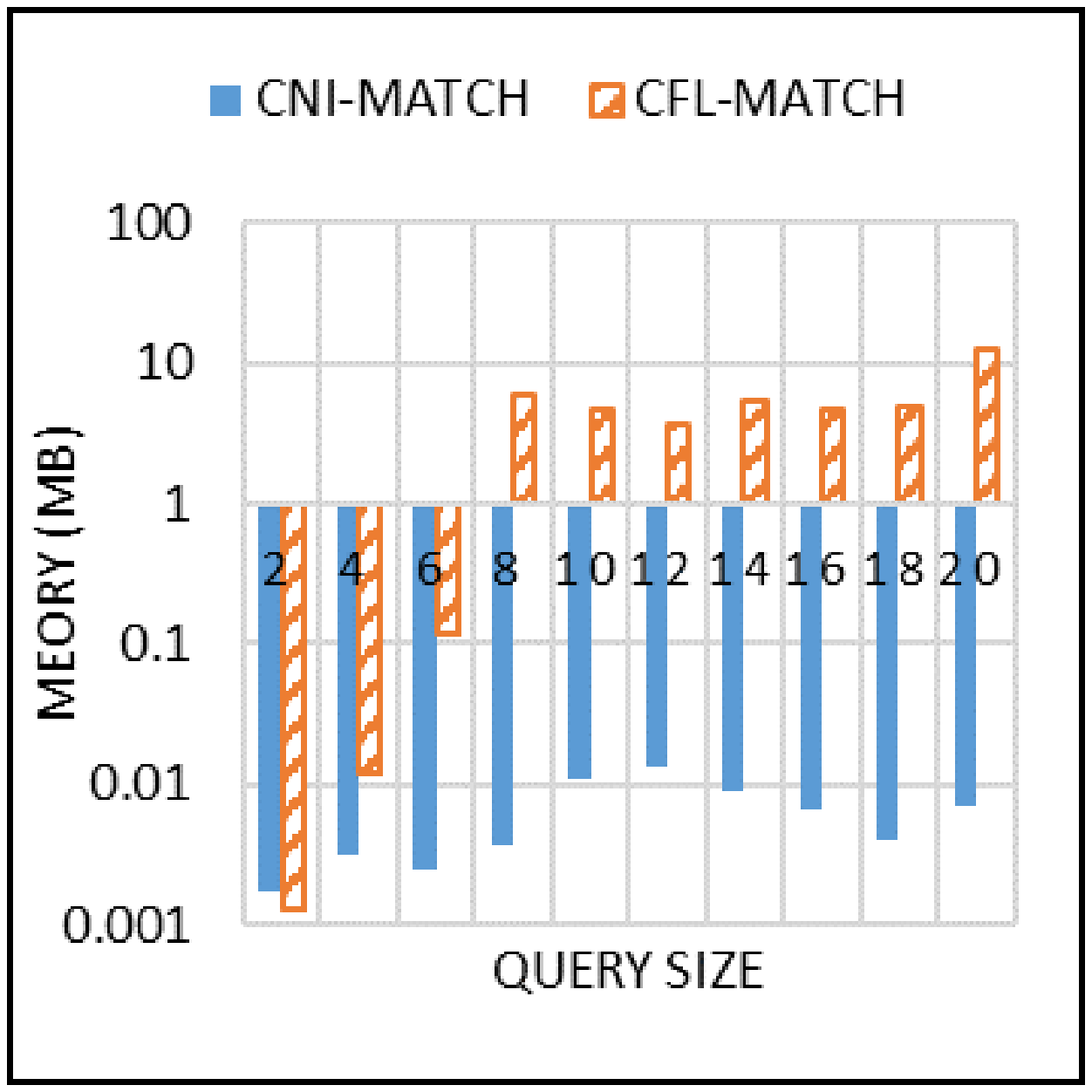}
\hspace{1cm}
\includegraphics[ scale=0.30]{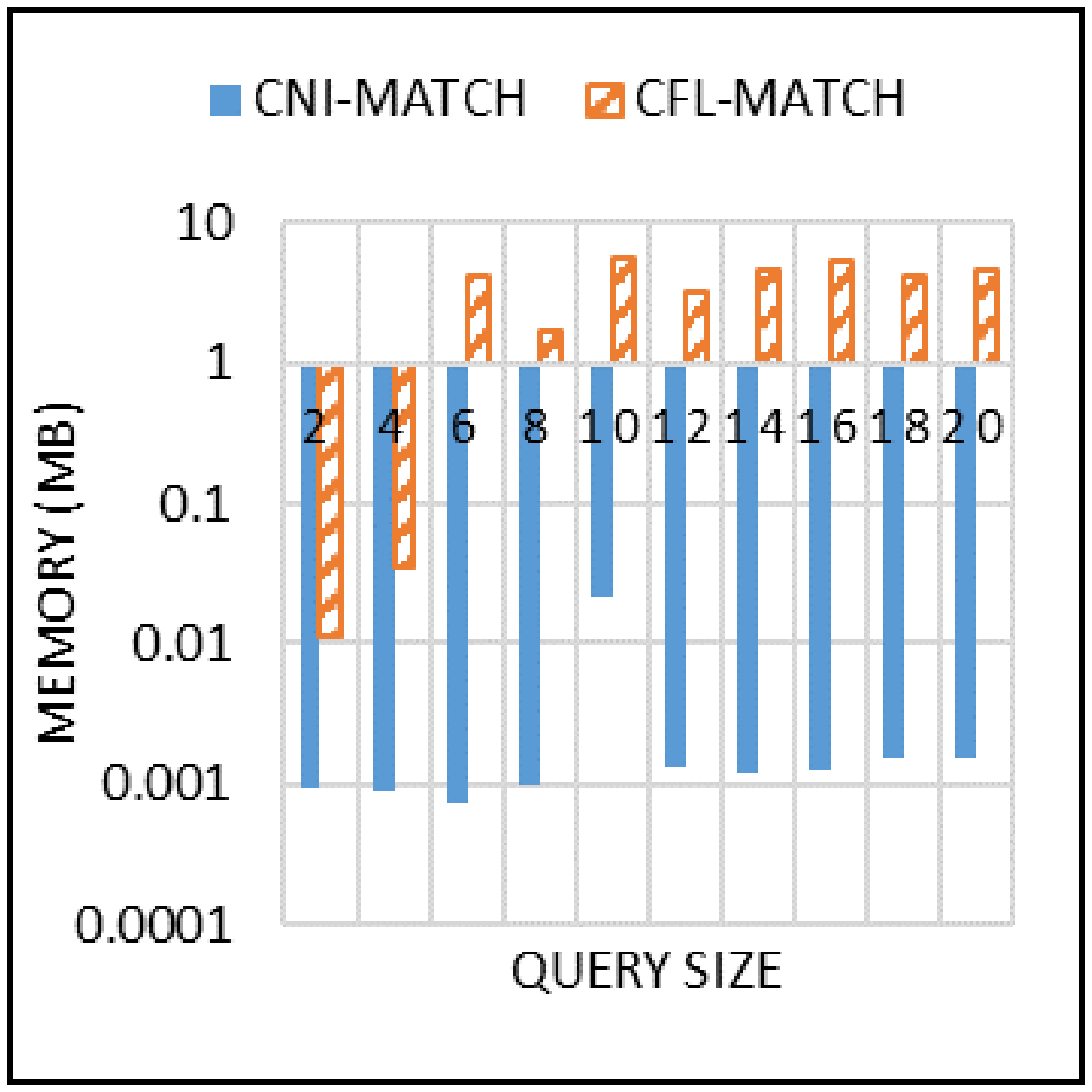}\\
(a) HUMAN dataset    \hspace{2cm}(b) YEAST dataset   \hspace{2cm}  (c) HPRD dataset
\caption{Memory space for sparse queries. (Results are in log scale). \label{Fig-SmallDS3}}
\end{figure*}
\begin{figure}[t]
\centering
\includegraphics[ scale=0.30]{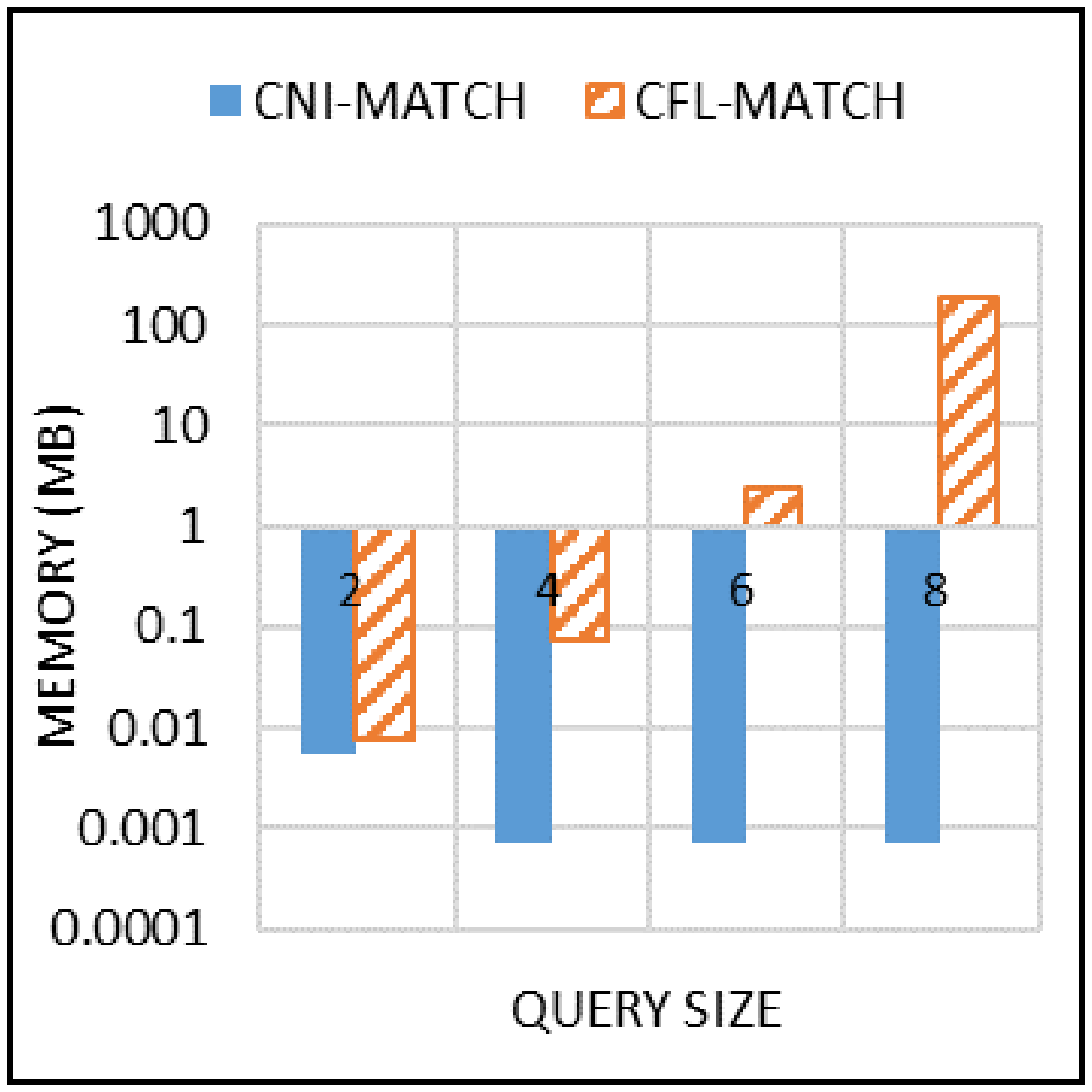}
\caption{Memory space for dense queries. \label{Fig-SmallDS4}}
(Results are in log scale)
\end{figure}
\begin{figure}[t]
\centering
\includegraphics[ scale=0.30]{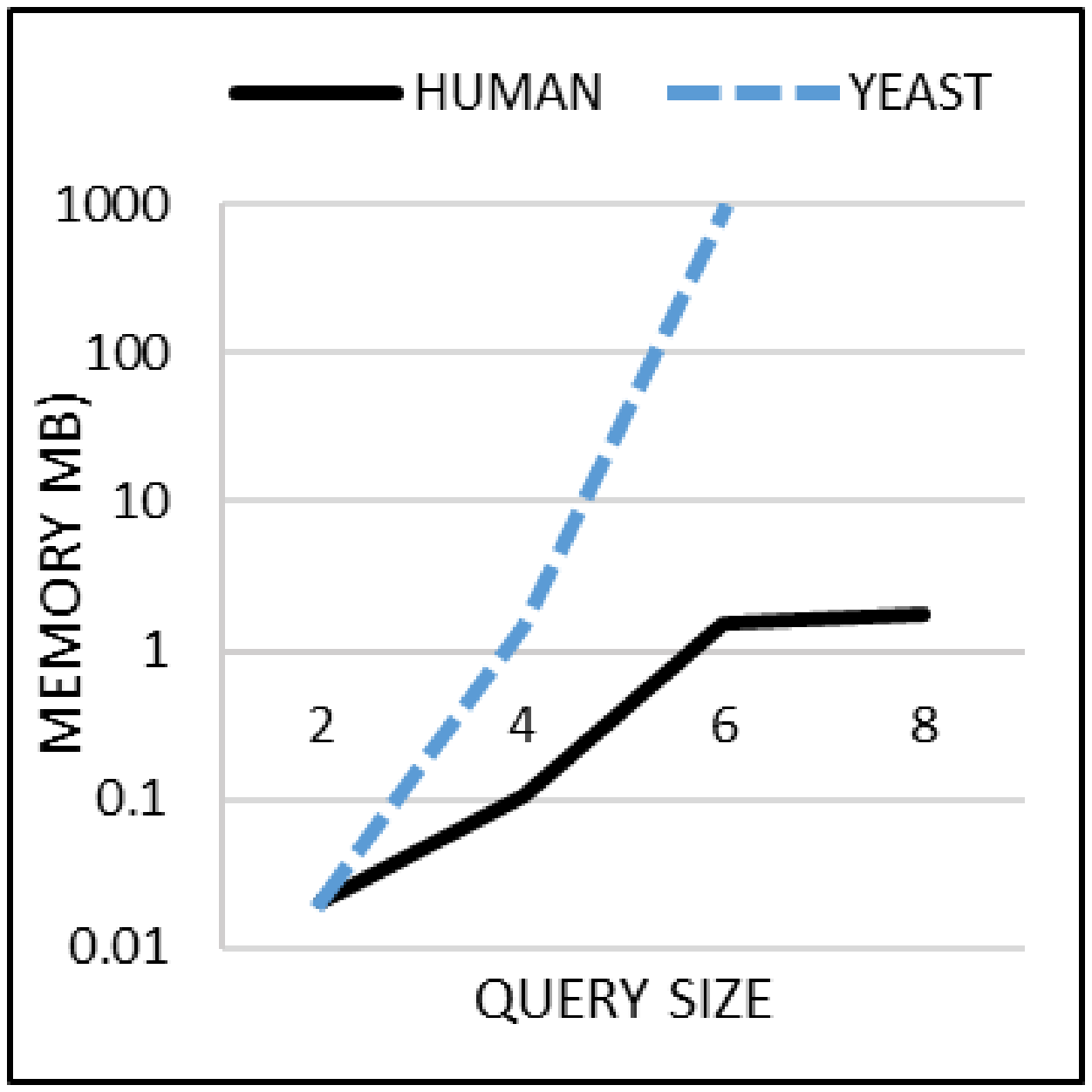}
\caption{Memory space for dense queries with CNI-Match. \label{Fig-SmallDS5}}
(Results are in log scale)
\end{figure}
To query the HUMAN, HPRD and YEAST datasets, we constructed a set of sparse and dense queries for each dataset. Each query is a connected subgraph of the data graph obtained using a random walk on the data
graph where the next vertex is selected according to the sparsity of the query. For a sparse query, the next vertex is selected among the neighbors that have the least number of neighbors. For a dense query, the next vertex is selected among the neighbors that have the greatest number of neighbors. For sparse queries, we provide $20$ query sets for each dataset, each containing
100 query graphs of the same size.
For dense queries, we provide $10$ query sets for each dataset, each containing
100 query graphs of the same size.

\subsection{Results}
In this subsection, we report and comment the results obtained by comparing our algorithm with the state of the art algorithm CFL-Match \cite{Bi2016}. Our main metric is the time performance by varying $|V(Q)|$, i.e., the number of vertices in the query, the density of the queries, and the amount of memory used by the algorithms. We present the obtained results according to these metrics. We note also that all the algorithms output the same sets of isomorphic subgraphs for each query graph.

Figure \ref{Fig-SmallDS} shows the average total processing time of the two algorithms on the three datasets when processing sparse queries. On the y-axis, INF means that the processing of the set of queries exceeded 12 hours execution time and has been aborted. According to this figure, CNI-Match is, compared to CFL-Match, on average 12 times faster on the YEAST dataset and 17 times faster on the HPRD dataset. For the dense and difficult dataset HUMAN, CNI-Match is 4 times faster than CFL-Match only by considering the query sizes for which CFL-Match has not reached the INF threshold.

Figure \ref{Fig-SmallDS2} shows the average total processing time of the two algorithms on the three datasets when processing dense queries. This figure shows clearly that CFL-Match is too slow for dense queries especially for the HUMAN dataset for which it has not obtained less than 12 hours even for the query size 2.

Figure \ref{Fig-SmallDS3} plots the memory space usage of the two algorithms when processing sparse queries for the 3 datasets. We can see in this figure that CNI-Match outperforms CFL-Match on the 3 datasets. CFL-Match which constructs the CPI structure for the data graph uses on average 675 times more space than CNI-Match for the YEAST dataset, and 2343 times more space for the HPRD dataset. For the HUMAN dataset, if we consider only the cases where CFL-Match finished within the INF threshold, it uses 4 times more space than CNI-Match. Note that these cases correspond to the sparse queries of size 1 to 14.

Figure \ref{Fig-SmallDS4} plots the memory space usage of the two algorithms when processing sparse queries for the HPRD dataset for which CFL-Match processed queries of size 1 to 8 within the INF threshold. CFL-Match space performance is very low as it uses on average 6000 times more space than CNI-Match.

Figure \ref{Fig-SmallDS5} shows how the amount of memory used by CNI-Match evolves when the query size increases for the HUMAN and the YEAST datasets. These amount of memory depends on the number of vertices of the data graphs that pass through the CNI filter and for which the algorithm stores the corresponding edges for adjacency testing.
\section{Conclusions}
\label{Sec-Conclusion}
Subgraph isomorphism search is an NP-complete problem. This means a processing time that grows with the size of the involved graphs. Pruning the search space is the pillar of a scalable subgraph isomorphism search algorithm and has been the main focus of proposed approaches since Ullmann's first solution. In this paper, we proposed \textit{CNI-Match}, a simple subgraph isomorphism search algorithm that relies on a compact representation of the neighborhood, called Compact Neighborhood Index (\textit{CNI}), to perform an early global pruning of the search space. \textit{CNI} distills topological information of each vertex into an integer. This vertex encoding is easily updatable and can be used to prune globally the search space using an iterative algorithm.  Our extensive experiments validate the efficiency of our approach.

As part of future work, it will be interesting to extend CNIs to construct a graph index that allows to handle a graph database.  For this issue, we propose to compute a vertex CNI that includes the vertex label:
$cni_d(u)=g_{2}(\ell(u), cni(u))$  where $\ell(u)$ is the label of vertex $u$ and $cni(u)$ its CNI at one hop. The CNI of the graph is given by:
$cni(G)=g_{|V(G)|}(cni_d(v_1),cni_d(v_2)+...+cni_d(v_{|V(G)|}))$ where  each $v_i$ is a vertex of $G$. This resulting graph CNI can be used to index a graph in a database of graphs defined on the same set of labels. This is not the sole manner to compute a CNI featuring a whole graph. So, it is interesting to study the soundness of such CNI in indexing by combining several approaches.
It is also interesting to note that computing \textit{CNI}s, does not require that the entire data graph is loaded into main memory. This is an interesting feature that can be used to deal with a graph stream. For this, we need to study the amount of edges that the algorithm needs to store according to the query size.

\bibliographystyle{abbrv}

\begin{thebibliography}{10}

\bibitem{Bachman1969}
C.~Bachman.
\newblock Data structure diagrams.
\newblock {\em DataBase: : A Quarterly Newsletter of SIGBDP}, 1(2), 1969.

\bibitem{Bi2016}
F.~Bi, L.~Chang, X.~Lin, L.~Qin, and W.~Zhang.
\newblock Efficient subgraph matching by postponing cartesian products.
\newblock In {\em Proceedings of the 2016 International Conference on
  Management of Data}, SIGMOD '16, pages 1199--1214, New York, NY, USA, 2016.
  ACM.

\bibitem{Bonnici2013}
V.~Bonnici, R.~Giugno, A.~Pulvirenti, D.~Shasha, and A.~Ferro.
\newblock A subgraph isomorphism algorithm and its application to biochemical
  data.
\newblock {\em BMC Bioinformatics}, 14(Suppl 7)(S13), 2013.

\bibitem{Cordella2004}
L.~P. Cordella, P.~Foggia, C.~Sansone, and M.~Vento.
\newblock {A (Sub)Graph Isomorphism Algorithm for Matching Large Graphs}.
\newblock {\em IEEE Transactions on Pattern Analysis and Machine Intelligence},
  26:1367--1372, 2004.

\bibitem{Gallagher2006}
B.~Gallagher.
\newblock Matching structure and semantics: A survey on graph-based pattern
  matching.
\newblock {\em AAAI FS}, 6:45--53, 2006.

\bibitem{Gallai1967}
T.~Gallai.
\newblock Transitiv orientierbare graphen.
\newblock {\em Acta Mathematica Hungarica}, 18:25--66, 1967.

\bibitem{habib2010}
M.~Habib and C.~Paul.
\newblock A survey of the algorithmic aspects of modular decomposition.
\newblock {\em Computer Science Review}, 4(1):41--59, 2010.

\bibitem{Han2013}
W.-S. Han, J.~Lee, and J.-H. Lee.
\newblock {Turboiso: Towards Ultrafast and Robust Subgraph Isomorphism Search
  in Large Graph Databases}.
\newblock In {\em Proceedings of the 2013 ACM SIGMOD International Conference
  on Management of Data}, SIGMOD '13, pages 337--348, New York, NY, USA, 2013.
  ACM.

\bibitem{He2006}
H.~He and A.~Singh.
\newblock Closure-tree: An index structure for graph queries.
\newblock In {\em Data Engineering, 2006. ICDE '06. Proceedings of the 22nd
  International Conference on}, pages 38--38, April 2006.

\bibitem{He2008}
H.~He and A.~K. Singh.
\newblock Graphs-at-a-time: Query language and access methods for graph
  databases.
\newblock In {\em Proceedings of the 2008 ACM SIGMOD International Conference
  on Management of Data}, SIGMOD '08, pages 405--418, New York, NY, USA, 2008.
  ACM.

\bibitem{Hopcroft2007}
J.~E. Hopcroft, R.~Motwani, and J.~D. Ullman.
\newblock {\em Introduction to Automata Theory, Languages, and Computation, 3rd
  Edition}.
\newblock Pearson, 2007.

\bibitem{Katsarou2017}
F.~Katsarou, N.~Ntarmoset, and P.~Triantafillou.
\newblock Subgraph querying with parallel use of query rewritings and
  alternative algorithms.
\newblock In {\em EDBT}, 2017.

\bibitem{Lagraa2014}
S.~Lagraa, H.~Seba, R.~Khennoufa, A.~M'Baya, and H.~Kheddouci.
\newblock A distance measure for large graphs based on prime graphs.
\newblock {\em Pattern Recognition}, 47(9):2993 -- 3005, 2014.

\bibitem{Lee2013}
J.~Lee, W.-S. Han, R.~Kasperovics, and J.-H. Lee.
\newblock An in-depth comparison of subgraph isomorphism algorithms in graph
  databases.
\newblock In {\em Proceedings of the 39th international conference on Very
  Large Data Bases}, PVLDB'13, pages 133--144. VLDB Endowment, 2013.

\bibitem{Lisi2007}
M.~Lisi.
\newblock Some remarks on the cantor pairing function.
\newblock {\em LE MATEMATICHE}, LXII-Fasc. I:55--65, 2007.

\bibitem{Nabti2017}
C.~Nabti and H.~Seba.
\newblock Querying massive graph data: A compress and search approach.
\newblock {\em Future Generation Computer Systems}, 74:63 -- 75, 2017.

\bibitem{Ren2015}
X.~Ren and J.~Wang.
\newblock Exploiting vertex relationships in speeding up subgraph isomorphism
  over large graphs.
\newblock {\em Proc. VLDB Endow.}, 8(5):617--628, Jan. 2015.

\bibitem{Fueter1923}
G.~P. Rudolf~Fueter.
\newblock Rationale abz\"{a}hlung der gitterpunkte, vierteljschr.
\newblock {\em Naturforsch. Ges, Z\"{u}rich}, 58:280–386, 1923.

\bibitem{Shang2008}
H.~Shang, Y.~Zhang, X.~Lin, and J.~X. Yu.
\newblock Taming verification hardness: An efficient algorithm for testing
  subgraph isomorphism.
\newblock {\em Proc. VLDB Endow.}, 1(1):364--375, Aug. 2008.

\bibitem{Shasha2002}
D.~Shasha, J.~T.~L. Wang, and R.~Giugno.
\newblock Algorithmics and applications of tree and graph searching.
\newblock In {\em Proceedings of the Twenty-first ACM SIGMOD-SIGACT-SIGART
  Symposium on Principles of Database Systems}, PODS '02, pages 39--52, New
  York, NY, USA, 2002. ACM.

\bibitem{Stein2013}
S.~K. Stein.
\newblock {\em Mathematics: The Man-Made Universe. New York: McGraw-Hill,
  1999}.
\newblock Dover Publications; 3rd Revised ed., March 21 2013.

\bibitem{Ullmann76}
J.~R. Ullmann.
\newblock {An Algorithm for Subgraph Isomorphism}.
\newblock {\em J. ACM}, 23(1):31--42, Jan. 1976.

\bibitem{Zhang2009}
S.~Zhang, S.~Li, and J.~Yang.
\newblock Gaddi: Distance index based subgraph matching in biological networks.
\newblock In {\em Proceedings of the 12th International Conference on Extending
  Database Technology: Advances in Database Technology}, EDBT '09, pages
  192--203, New York, NY, USA, 2009. ACM.

\bibitem{ZhaoP2010}
P.~Zhao and J.~Han.
\newblock On graph query optimization in large networks.
\newblock {\em {PVLDB}}, 3(1):340--351, 2010.

\bibitem{Zhao2012}
X.~Zhao, C.~Xiao, X.~Lin, and W.~Wang.
\newblock {Efficient Graph Similarity Joins with Edit Distance Constraints}.
\newblock In {\em IEEE 28th International Conference on Data Engineering (ICDE
  2012), Washington, DC, USA (Arlington, Virginia), 1-5 April}, pages 834--845,
  2012.

\bibitem{Zhu2012}
G.~Zhu, X.~Lin, K.~Zhu, W.~Zhang, and J.~X. Yu.
\newblock Treespan \: Efficiently computing similarity all-matching.
\newblock In {\em Proceedings of the 2012 ACM SIGMOD International Conference
  on Management of Data}, SIGMOD '12, pages 529--540, New York, NY, USA, 2012.
  ACM.

\end{thebibliography}

\appendix
\section{Proof of Theorem 1}
\label{app-prooftheo}
\begin{proof}
We need the following lemmas.
\begin{lem}\label{lemme1}
 $p<p' \Rightarrow \hbar(k,p)< \hbar(k,p')$
 \end{lem}
 \begin{proof}
 By deduction from the property of the binomial coefficient: $\binom{n}{k}=\binom{n-1}{k}+\binom{n-1}{k-1}$ (Pascal Formula)
\end{proof}
\begin{lem}\label{lemme2}
 $\forall k>0, g_{k}(x_{1},...,x_{k})<\hbar(k,x_{1}+...+x_{k}+1)$
  \end{lem}
 \begin{proof}
 This inequality is trivial for $k=1$: $g_1(x_1)=x_1$ and $\hbar(1,x_1+1)=x_1+1$.
 Assume that, for $k\geq1$, the inequality holds and let us prove that it also holds for $k+1$.\\
 By definition of $g_k$, we have:\\
 $g_{k+1}(x_{1},...,x_{k+1})=g_{k}(x_{1},...,x_{k})+\hbar(k+1,x_{1}+...+x_{k+1})$
 $\qquad{}\qquad{}\qquad{}\qquad{}<\hbar(k,x_{1}+...+x_{k}+1)+\hbar(k+1,x_{1}+...+x_{k+1})$
 $\qquad{}\qquad{}\qquad{}\qquad{}<\hbar(k,x_{1}+...+x_k+x_{k+1}+1)+\hbar(k+1,x_{1}+...+x_{k+1})$

By the property of Pascal's triangle, we know that:\\ $\hbar(k,x_{1}+...+x_k+x_{k+1}+1)+\hbar(k+1,x_{1}+...+x_{k+1})=\hbar(k+1,x_{1}+...+x_{k+1}+1)$,
we have $g_{k+1}(x_{1},...,x_{k+1})<\hbar(k+1,x_{1}+...+x_{k+1}+1)$

\end{proof}

\begin{lem}\label{lemme3}
 $\forall k>0$, If $g_{k}(x_{1},...,x_{k})= g_{k}(x'_{1},...,x'_{k})$ then $ x_{1}+...+x_{k}= x'_{1}+...+x'_{k}$
  \end{lem}
 \begin{proof}
  Assume that $g_{k}(x_{1},...,x_{k})= g_{k}(x'_{1},...,x'_{k})$. \\
  According to Lemma \ref{lemme2}, we have:\\

  $ \hbar(k,x_{1}+...+x_{k})<g_{k}(x_{1},...,x_{k})= g_{k}(x'_{1},...,x'_{k})<\hbar(k,x_{1}+...+x_{k}+1$

  we obtain then: $ \hbar(k,x_{1}+...+x_{k})<\hbar(k,x_{1}+...+x_{k}+1$
  According to Lemma \ref{lemme1}, $\hbar(k,p)$ is strictly increasing. So, the inequality $x_{1}+...+x_{k}\leq x'_{1}+...+x'_{k}$ holds.
  Similarly, we prove the inverse inequality.  This proves that $x_{1}+...+x_{k}= x'_{1}+...+x'_{k}$.

\end{proof}

To prove Theorem \ref{Theo-bijection}, we first prove that $g_k$ is  injective from $\mathds{N}^k$ to $\mathds{N}$. It is trivial for $k=1$. In fact, $g_1=\hbar(1,x_1)=\binom{x_1}{1}=\frac{x_1!}{1!(x_1-1)!}=x_1$ is the identity in $\mathds{N}$.
For $k\geq2$, we assume that $g_{k-1}$ is injective and we prove that $g_{k}$ is also injective.
Let $(x_{1},....,x_{k})$ and $(x'_{1},....,x'_{k})$ such that $g_{k}(x_{1},....,x_{k})=g_{k}(x'_{1},....,x'_{k})$.
According to Lemma \ref{lemme3}, $x_{1}+...+x_{k}= x'_{1}+...+x'_{k}$. We have also by definition of $g_k$:\\
$$
\left\{ \begin{array}{l}
g_{k}(x_{1},...,x_{k})=g_{k-1}(x_{1},...,x_{k-1})+\hbar(k,x_{1}+...+x_{k}) \\
g_{k}(x'_{1},...,x'_{k})=g_{k-1}(x'_{1},...,x'_{k-1})+\hbar(k',x_{1}+...+x'_{k})\\
\end{array}
\right.
$$
By subtracting side by side, we obtain $g_{k-1}(x_{1},...,x_{k-1})=g_{k-1}(x'_{1},$ $...,x'_{k-1})$ which is our induction hypothesis that gives $(x_{1},...,x_{k-1})=(x'_{1},...,x'_{k-1})$. This implies that $x_k=x'_k$.

Conclusion: $g_{k}$ is injective.

To show that $g_k$ is also surjective, we recall that
$\hbar(k,x_{1}+...+x_{k})\leq g_{k}(x_{1},...,x_{k})<\hbar(k,x_{1}+...+x_{k}+1)$.
As  $p\rightarrow \hbar(k,p)$ is a strictly increasing sequence, we deduce that
each $n$ in $\mathds{N}$ have an antecedent in $\mathds{N}^{k}$.\\
So, $g_{k}$ is a bijection from $\mathds{N}^{k}$ to $\mathds{N}$ which proves Theorem \ref{Theo-bijection}.

\end{proof}
\section{Proof Sketch of Lemma \ref{lem-cni}}
\label{app-prooflemcni}
We prove the lemma by contradiction. Assume $v$ is a candidate of $u$ with $cni(v) < cni(u)$. That is, there is an embedding $M$ that maps $u$ to $v$. This means that $\ell(v)=\ell(u)$ and $deg(v)\geq deg(u)$ and $\ell(N(u))\subseteq \ell(N(v))$.  Let $deg(u)=k$ and $deg(v)=k+t$, $t\geq1$. Let $(l_1, l_2, \cdots, l_k)$ be the labels of the neighbors of $u$ according to the order given by function $ord()$. Similarly, let $(l_1, l_2, \cdots, l_k, l_{k+1}, \cdots, l_{k+t})$ be the labels of the neighbors of $v$. By construction of, we have
$cni(v)=g_{k+t}(l_1, l_2, \cdots, l_{k+t})$=$g_k(l_1, l_2,$ $\cdots, l_{k})$+$\hbar(k+1,l_{1}+...+l_{k+1})$+$\cdots$+$\hbar(k+t,l_{1}+$ $...+l_{k+t})$.
So, $cni(v)=cni(u)$+$\hbar(k+1,l_{1}+...+l_{k+1})$+$\cdots$+$\hbar(k+t,l_{1}+...+l_{k+t})$.
as $t>0$, we reach a contradiction. Thus, the lemma holds.

\end{document}